\documentclass{article}
\usepackage[english]{babel}
\usepackage{graphicx} 
\usepackage[letterpaper,top=2cm,bottom=2cm,left=3cm,right=3cm,marginparwidth=1.75cm]{geometry}
\usepackage[colorlinks,linkcolor=blue,citecolor=red,urlcolor=blue]{hyperref}
\usepackage{upgreek}
\usepackage{xspace}
\usepackage{amsmath, amssymb, amsthm}
\usepackage[ruled]{algorithm}
\usepackage{algorithmic}
\usepackage{authblk}
\usepackage{natbib}
\usepackage{multibib}
\newcites{app}{References for Supplementary Material}

\def\real{\mathbb{R}}
\def\cP{{\mathcal P}}

\def\N{{\mathcal N}}

\def\cG{{\mathcal G}}

\DeclareMathOperator*{\dom}{dom}
\newcommand{\KL}{\mathop{\mathrm{KL}}\nolimits}
\newcommand{\cF}{\mathcal F}
\newcommand{\X}{\mathsf E}
\newcommand{\calX}{\mathcal X}
\newcommand{\ssp}[1]{\langle #1 \rangle}
\DeclareMathOperator*{\argmin}{argmin}
\newcommand{\norm}[1]{\ensuremath{\Vert #1 \Vert}}

\newcommand{\bounded}[1][\calX]{\ensuremath{\mathcal{B}_b(#1)}}
\def\weight{G_n}
\def\weightN{G_n^N}

\def\predictiveN{\eta^N_n}
\def\predictive{\eta_n}
\def\updateN{\hat{\eta}_n^N}
\def\update{\hat{\eta}_n}
\newcommand{\bgN}{\ensuremath{\Psi_{G_n^N}}\xspace}
\newcommand{\bg}{\ensuremath{\Psi_{G_n}}\xspace}

\DeclareMathOperator{\Exp}{\mathbb{E}}
\newcommand{\supnorm}[1]{\norm{#1}_{\infty}}
\newcommand{\testfn}{\ensuremath{\varphi} }
\def\sfresampling{\mathcal{F}_{n-1}^N}
\def\sfmutation{\mathcal{G}_{n-1}^N}

\newtheorem{definition}{Definition}
\newtheorem{example}{Example}
\newtheorem{prop}{Proposition}
\newtheorem{cor}{Corollary}
\newtheorem{lemma}{Lemma}
\newtheorem{remark}{Remark}
\newtheorem{assumption}{Assumption}

\usepackage[textwidth=1.8cm, textsize=scriptsize]{todonotes}

\title{A mirror descent approach to maximum likelihood estimation in latent variable models}
\author[1]{Francesca Romana Crucinio\thanks{ \href{mailto:francescaromana.crucinio@unito.it}{francescaromana.crucinio@unito.it}
    The author gratefully acknowledges the ``de Castro" Statistics Initative at the \textit{Collegio Carlo Alberto} and the \textit{Fondazione Franca e Diego de Castro}. The author is  supported by the Gruppo
Nazionale per l'Analisi Matematica, la Probabilità e le loro Applicazioni (GNAMPA-INdAM).}}
\date{ }

\affil[1]{ESOMAS, University of Turin, Italy \& Collegio Carlo Alberto, Turin, Italy}
\begin{document}

\maketitle

\graphicspath{{Images/}}

\begin{abstract}
    We introduce an approach based on mirror descent and sequential Monte Carlo (SMC) to perform joint parameter inference and posterior estimation in latent variable models. This approach is based on minimisation of a functional over the parameter space and the space of probability distributions and, contrary to other popular approaches, can be implemented when the latent variable takes values in discrete spaces. We provide a detailed theoretical analysis of both the mirror descent algorithm and its approximation via SMC. We experimentally show that the proposed algorithm outperforms standard expectation maximisation algorithms and is competitive with other popular methods for real-valued latent variables.
\end{abstract}

\section{Introduction}

Parameter inference for latent variable models (LVMs) is a classical task in statistical learning. These models are flexible and can describe the hidden structure of complex data such as images \citep{bishop2006pattern}, text \citep{blei2003latent}, audio \citep{smaragdis2006probabilistic}, and graphs \citep{hoff2002latent}.
LVMs are probabilistic models with observed data $y$ and likelihood
$p_\theta(x, y)$ parametrised by $\theta\in \real^{d_\theta}$, where $x\in\calX$ is a latent variable which cannot be observed.
In the frequentist inference framework, the interest is in estimating the parameter through maximisation of the marginal likelihood of the observed data 
\begin{align}
\label{eq:mmle}
\theta^\star \in \arg \max_{\theta \in \real^{d_\theta}} p_\theta(y) = \arg \max_{\theta \in \real^{d_\theta}} \int p_\theta(x, y) d x,
\end{align}
an approach often called maximum marginal likelihood estimation (MMLE).
A pragmatic compromise between frequentist
and Bayesian approaches, is the empirical Bayes paradigm in which the MMLE is complemented by uncertainty estimation over the latent variables $x$ via the posterior $p_\theta(x|y)= p_\theta(x, y)/p_\theta(y)$. These two tasks are intertwined, and often estimation of $\theta^\star$ and of the corresponding posterior need to be performed simultaneously.

The standard approach for solving~\eqref{eq:mmle} is the expectation-maximisation (EM) algorithm \citep{dempster1977maximum}, which was first proposed in the context of missing data. EM proceeds iteratively by alternating an expectation step with respect to the latent variables and a maximisation step with respect to the parameter. The expectation step requires knowledge of the posterior $p_\theta(\cdot|y)$ while the maximisation step assumes that a surrogate of $p_\theta(y)$ can be maximised analytically.
The wide use of the EM algorithm is due to the fact that it can be implemented using approximations for both steps: analytic maximisation can be replaced by numerical optimisation \citep{meng1993maximum, liu1994ecme} and the expectation step can be approximated via Monte Carlo sampling from $p_\theta(\cdot|y)$ \citep{wei1990monte, celeux1985sem}. When exact sampling from the posterior is unfeasible, approximate samples can be drawn via Markov chain Monte Carlo (MCMC; \cite{de2021efficient, delyon1999convergence}) leading to stochastic approximation EM (SAEM). 

Recently, \cite{kuntz2023particle} explored an approach based on \cite{neal1998view} which shows that the EM algorithm is equivalent to performing coordinate descent of a free energy functional, whose minimum is the maximum likelihood estimate. They propose several interacting particle algorithms to address the optimisation problem. 
An alternative approach to MMLE is to define a distribution which concentrates on $\theta^\star$; this can be achieved borrowing techniques from simulated annealing (see, e.g., \cite{van1987simulated}). Any Monte Carlo method sampling from such a distribution would then approximate $\theta^\star$: \cite{gaetan2003multiple, jacquier2007mcmc, doucet2002marginal} consider MCMC algorithms, \cite{johansen2008particle} sequential Monte Carlo algorithms and \cite{akyildiz2023interacting} unadjusted Langevin algorithms.

Our work also stems from the observation that maximising $p_\theta(y)$ is equivalent to minimising a certain functional $\cF$ over the product space $\real^{d_\theta}\times \cP(\calX)$, but, contrary to \cite{kuntz2023particle} which apply a gradient descent strategy to minimise $\cF$ and obtain an algorithm based on a pair of stochastic differential equations, we consider a mirror descent approach which does not require knowledge of $\nabla_x \log p_\theta(x, y)$. As a consequence our method can be applied in settings in which the joint likelihood is not differentiable in $x$ and in cases in which $\calX$ is a discrete space. This is often the case in mixture models and models for graphs.

Leveraging the connection between mirror descent established in \cite{chopin2023connection} and sequential Monte Carlo (SMC) algorithms we propose an SMC method to perform joint parameter inference and posterior approximation in LVMs. We also consider a second SMC approximation to speed up computation time. We provide a theoretical analysis of the developed methods and obtain precise error bounds for the parameter.
We consider a wide range of toy and challenging experiments and compare with EM and its variants as well as methods sampling from a distribution concentrating on $\theta^\star$. Compared to EM, methods based on minimisation of the functional $\cF$ (as ours) and on simulated annealing suffer less with issues related to local maximisers. Compared to other methods based on minimisation of $\cF$, our approach does not require differentiability in $x$.

This paper is organised as follows.
In Section~\ref{sec:md}, we introduce mirror descent and its adaptation for the MMLE problem. In Section~\ref{sec:smc}, we introduce the necessary background on SMC, describe two numerical approximations of mirror descent for MMLE via SMC and provide their theoretical analysis. 
In Section~\ref{sec:ex}, we show comprehensive numerical experiments comparing the results obtained with our method with EM and other competitors. We conclude with Section~\ref{sec:conclusion}. Code to reproduce all experiments is available at \url{https://github.com/FrancescaCrucinio/MD_LVM}. Proofs of all results and additional experimental details can be found in the Supplement.

\subsection*{Notation}
Let $\X$ be a topological vector space endowed with the Borel $\sigma$-field $\mathcal{B}(\X)$. We denote by $\X^*$ the dual of $\X$ and for any $z\in\X$ and $z^\star\in\X^*$ we denote the dot product by $\ssp{z^\star, z}$.
We denote by $\cP(\X)$ the set of probability measures on $\X$. The Kullback--Leibler divergence is defined as follows: for $\nu,\mu \in \cP(\X)$, $\KL(\nu|\mu)=\int \log(\frac{d\nu}{d\mu}) d\nu$ if $\nu$ is absolutely continuous with respect to $\mu$ with Radon-Nikodym density $\frac{d\nu}{d\mu}$, and $+\infty$ else.

We denote by $\mathcal{N}(x; m, \Sigma)$ the density of a multivariate Gaussian with mean $m$ and covariance $\Sigma$ and by $\textsf{Unif}(a, b)$ the uniform distribution over $[a, b]$.
Throughout the manuscript we assume $\theta\in \real^{d_\theta}$, $y\in \real^{d_y}$ and $x\in\calX$ with $\calX\subseteq \real^{d_x}$ or $\calX\subseteq \mathbb{Z}^{d_x}$.
\section{A mirror descent approach to maximum likelihood}
\label{sec:md}

\subsection{Background on mirror descent}

Let $\cF:\X\rightarrow \real^+$ be a functional on $\X$ and consider the optimisation problem $\min_{z \in \X} \cF(z)$.
Mirror descent \cite{nemirovsky1983problem} is a first-order optimisation scheme relying on the derivatives of the objective functional, and a geometry on the search space induced by Bregman divergences.

\begin{definition}[Derivative]
\label{def:first_var}
If it exists, the derivative of $\cF$ at $z_1$ is the  element $\nabla \cF(z_1) \in\X^\star $ s.t. for any $z_2 \in \X$, with $\xi = z_2-z_1$:
\begin{equation*}
\lim_{\epsilon \rightarrow 0}\frac{1}{\epsilon}(\cF(z_1+\epsilon  \xi) -\cF(z_1))
=\ssp{\nabla\cF(z_1), \xi},
\end{equation*} 
and is defined uniquely up to an additive constant.
\end{definition}
If $\X=\real$, then this notion of derivative coincides with the standard one, while if $\X=\mathcal{P}(\real^d)$ this corresponds to the first variation of  $\cF$ at $z_1$.

\begin{definition}[Bregman divergence]
Let $\phi:\X\rightarrow \real^+$ be a convex and continuously differentiable functional on $\X$. The $\phi$-Bregman divergence is defined for any $z_1,z_2 \in \X$ by:
\begin{equation}
\label{eq:breg_probas}
     B_{\phi}(z_1|z_2)=\phi(z_1)-\phi(z_2)-\ssp{\nabla \phi(z_2), z_1- z_2},
\end{equation}
where $\nabla \phi(z_2)$ denotes the derivative of $\phi$ at $z_2$.
\end{definition}

Given an initial $z_0\in \X$ and a sequence of step-sizes $(\gamma_n)_{n\ge 1}$, one can generate the sequence $(z_{n})_{n \geq 0}$  as follows
\begin{align}\label{eq:prox_algo}
    z_{n+1}=\argmin_{z \in \X}\left\{ \cF(z_n) +\ssp{ \nabla \cF(z_n), z-z_n}+ (\gamma_{n+1})^{-1} B_{\phi}(z|z_n) \right\}.
\end{align}
Writing the first order conditions of \eqref{eq:prox_algo}, we obtain the dual iteration
\begin{equation}   \label{eq:dual_iter} \nabla \phi(z_{n+1} ) -  \nabla \phi(z_n) = - \gamma_{n+1} \nabla \cF(z_n).
\end{equation}

\begin{remark}
If $\X=\real^d$ and $B_\phi = \norm{\cdot}^2/2$, then~\eqref{eq:dual_iter} is equivalent to the standard gradient descent algorithm. Let $\X=\cP(\real^d)$, $\cF(\mu)=\KL(\mu|\pi)$, \cite{chopin2023connection} shows that mirror descent with $B_{\phi}(\pi|\mu)=\KL(\pi|\mu)$ leads to the tempering (or annealing) sequence: for $\lambda_n = 1-\prod_{k=1}^n(1-\gamma_k)$,
\begin{equation}\label{eq:tempering}
  \mu_{n+1}\propto\mu_n^{(1-\gamma_{n+1})}\pi^{\gamma_{n+1}}=\mu_0^{\prod_{k=1}^n(1-\gamma_k)}\pi^{1- \prod_{k=1}^n(1-\gamma_k)} = \mu_0^{1-\lambda_n}\pi^{\lambda_n}.
\end{equation}
\end{remark}

\subsection{Mirror descent for maximum marginal likelihood estimation}
Let $y\in\real^{d_y}$ denote the observed data, $x\in\calX$ the latent variables and $\theta\in\real^{d_\theta}$ the parameter of interest. 
The goal of maximum marginal likelihood estimation (MMLE) is to find the parameter $\theta^\star$ that maximises the marginal likelihood $p_\theta(y)= \int p_\theta(x, y) d x$.

\cite{neal1998view, kuntz2023particle} show that minimisation of the functional $\mathcal{F}:\real^{d_\theta}\times\cP(\calX)\to \mathbb{R}$, defined by
\begin{align}
\label{eq:f}
    \mathcal{F}(\theta, \mu) &:= \KL(\mu| p_\theta(\cdot, y))
    = \int  U(\theta, x)\mu(x)dx + \int \log\left( \mu(x)\right)\mu( x)dx,
\end{align}
where we defined the negative log-likelihood as $U(\theta, x) := - \log p_\theta(x, y)$ for any fixed $y\in \real^{d_y}$, is equivalent to marginal maximum likelihood estimation.

In the following we assume that all probability measures considered admit a density w.r.t a dominating measure (e.g. the Lebesgue measure if $\calX\subseteq\real^{d_x}$ or the counting measure if $\calX\subseteq\mathbb{Z}^{d_x}$).
We also assume that $(\theta, x)\mapsto U(\theta, x)$ is sufficiently regular for $\nabla_\theta U(\theta, x)$ to be well-defined and that Leibniz integral rule for differentiation under the integral sign
(e.g. \cite[Theorem 16.8]{billingsley1995measure}) allows us to swap gradients with integrals.


In order to define a mirror descent scheme for $\cF$, we need an appropriate notion of derivative and Bregman divergence (see Appendix~\ref{app:ingredients_md} for a proof). 
\begin{prop} 
\label{prop:derivative}
\begin{enumerate}
\item Recall that for a functional $\cF$ the derivative $\nabla \cF$ is the element of the dual such that
$\lim_{\epsilon \rightarrow 0}\epsilon^{-1}(\cF(z_1+\epsilon  \xi) -\cF(z_1))
=\ssp{\nabla\cF(z_1), \xi}$. The derivative of $\cF$ in~\eqref{eq:f} is given by $\nabla \cF:\real^d\times\cP(\calX)\to \real^d\times\real$ where $$\nabla \cF(\theta, \mu) = \begin{pmatrix}\int \nabla_\theta  U(\theta, x)\mu( x)dx\\ \log\mu(x)+ U(\theta, x) )\end{pmatrix}.$$
\item Let $B_h$ be any Bregman divergence over $\real^{d_\theta}$ and denote $z = (\theta, \mu)$. Then $B_\phi(z_1|z_2) = B_h(\theta_1| \theta_2)+\KL(\mu_1|\mu_2)$ is a Bregman divergence over $\real^{d_\theta}\times\cP(\calX)$ given by the potential $\phi:(\theta, \mu)\mapsto \int \log (\mu(x))\mu(x)dx + h(\theta)$.
\end{enumerate}

\end{prop}
Given that $\nabla\phi(\mu) = \log \mu$, plugging the above into~\eqref{eq:dual_iter} we obtain
\begin{align*}
    \nabla h(\theta_{n+1}) - \nabla h(\theta_{n}) &= -\gamma_{n+1}
        \int \nabla_\theta U(\theta_n, x)\mu_n(x)dx\\
        \log (\mu_{n+1}(x)) - \log(\mu_n(x)) &=-\gamma_{n+1}\left[\log(\mu_n(x)) + U(\theta_n, x)\right].
\end{align*}
Exponentiating the second component and since $\nabla h$ is bijective due to the convexity of $h$, we can write the following updates:
\begin{align}
\label{eq:md_update_mmle}
    \theta_{n+1} &= (\nabla h)^{-1}\left(\nabla h(\theta_{n}) - \gamma_{n+1}\int \nabla_\theta U(\theta_n, x)\mu_n(x)dx\right)\\
    \mu_{n+1}(x) &\propto \mu_n(x)^{(1-\gamma_{n+1})}p_{\theta_n}(x, y)^{\gamma_{n+1}}, \label{eq:md_update}
\end{align}
which corresponds to a standard mirror descent step in $\real^d$ for $\theta$ and an update in the space of probability measures for $\mu$. However, contrary to~\eqref{eq:tempering}, the target distribution $\pi\equiv p_{\theta_n}$ changes at every iteration.
The equations above lead to an iterative scheme to perform MMLE which only requires the derivative of $p_\theta$ with respect to $\theta$, contrary to the schemes proposed in \cite{kuntz2023particle, akyildiz2023interacting}, which minimise the same functional. In addition, as we show in Section~\ref{sec:smc}, the iteration over $\mu_n$ can be efficiently implemented via sequential Monte Carlo (see, e.g., \cite{chopin2020introduction}).

Following the approach of \cite{lu2018relatively, aubin2022mirror} we can obtain an explicit convergence result for the scheme~\eqref{eq:md_update_mmle} under the following assumptions:

\begin{assumption}
\label{ass:convex}
    Assume that $\theta \mapsto U(\theta, x)$ is $l$-relatively convex with respect to $h$ uniformly in $x$, that is, for some $l\geq 0$
\begin{align*}
     U(\theta_2, x) \geq  U(\theta_1, x)+\ssp{\nabla_\theta    U(\theta_1, x), \theta_2-\theta_1}+l B_h(\theta_2|\theta_1),
\end{align*}
and $L$-relatively smooth with respect to $h$ uniformly in $x$, for some $L>0$, that is,
\begin{align*}
     U(\theta_2, x) \leq  U(\theta_1, x)+\ssp{\nabla_\theta    U(\theta_1, x), \theta_2-\theta_1}+L B_h(\theta_2|\theta_1).
\end{align*}
\end{assumption}
Relative smoothness is a weaker condition than gradient-Lipschitz continuity and relative strong convexity implies standard strong convexity since $B_h(\theta_2|\theta_1) \geq \norm{\theta_2-\theta_1}^2/2$ \citep{lu2018relatively}.
These assumptions are similar to those of \cite{akyildiz2023interacting,caprio2024error}; however, in our case, we can limit ourselves to uniform convexity and smoothness in $\theta$ and do not need a similar assumption on the $x$ component.

The proof of the following proposition is given in Appendix~\ref{app:ingredients_md} and follows along the lines of that of \citet[Proposition 1]{chopin2023connection}.
\begin{prop}[Convergence of Mirror Descent]
\label{cor:convergence}
    Let $(\theta_0, \mu_0) \in \real^d\times\cP(\calX)$ be an initial pair of parameter and distribution. Denote by $\theta^\star$ the MMLE and by $p_{\theta^\star}(\cdot|y)$ the corresponding posterior distribution. Under Assumption
    ~\ref{ass:convex} and if $\gamma_n\leq \min(l, 1, L^{-1})^{-1}$ for all $n\geq 1$, we have
\begin{align*}
0&\leq \mathcal{F}(\theta_n, \mu_n) - \log p_\theta^\star(y)\le (\gamma_1)^{-1} \prod_{k=1}^n(1-\gamma_k\min(l, 1)) [\KL(p_{\theta^\star}(\cdot|y)|\mu_0)+B_h(\theta^\star| \theta_0)] \overset{n\to \infty}{\to} 0.
\end{align*}
\end{prop}
Proposition~\ref{cor:convergence} guarantees that the iterates~\eqref{eq:md_update_mmle} converge the the minimiser of $\cF$.
Due to the nature of mirror descent not requiring differentiability in $x$ of $U$, our results apply to both $\calX\subseteq \real^{d_x}$ and discrete spaces such as $\calX\subseteq \mathbb{Z}^{d_x}$. 
\section{Sequential Monte Carlo for mirror descent}
\label{sec:smc}
\subsection{Background on SMC}
Sequential Monte Carlo (SMC) samplers \citep{del2006sequential} provide particle approximations of a sequence of distributions $(\mu_n)_{n= 0}^T$ using clouds of $N$ weighted particles $\{X_n^i, W_n^i\}_{i=1}^N$. 
To build an SMC sampler one needs the sequence of distributions $(\mu_n)_{n= 0}^T$, a family of Markov kernels $(M_n)_{n= 1}^T$ and a resampling scheme.
The sequence of distributions $(\mu_n)_{n= 0}^T$ is 
chosen to bridge an easy-to-sample from $\mu_0$ (e.g. the prior) to the target of interest $\mu_T=\pi$ (e.g. the posterior). A popular choice is
~\eqref{eq:tempering} with $0=\lambda_0\leq \dots \leq\lambda_T=1$. In this case, SMC samplers provide an approximation to mirror descent.

Starting from $\{X_0^i, W_0^i\}_{i=1}^N$ approximating $\mu_0$, SMC evolves the particle cloud to approximate the sequence of distributions $(\mu_n)_{n= 0}^T$ by sequentially updating the particle locations via the Markov kernels $(M_n)_{n= 1}^T$ and reweighting the newly obtained particles using a set of weights. After reweighting the particles are resampled. Broadly speaking, a resampling scheme is a selection mechanism which given a set of weighted samples $\lbrace X_n^i, W_n^i\rbrace_{i=1}^N$ outputs a sequence of equally weighted samples $\lbrace \widetilde{X}_n^i, 1/N\rbrace_{i=1}^N$ in which $\widetilde{X}_n^i = X_n^j$ for some $j$ for all $i=1,\ldots, N$. For a review of commonly used resampling schemes see \cite{Gerber2019}. At $n=T$, the particles approximate $\mu_T=\pi$ (see Appendix~\ref{app:proof} for the algorithm).

For simplicity, we focus on the case in which the Markov kernels $M_n$ are $\mu_{n-1}$-invariant. This choice allows, under some conditions \citep[3.3.2.3]{del2006sequential}, to obtain the following expression for the importance weights to move from distribution $\mu_{n-1}$ to $\mu_n$
\begin{align}\label{eq:smc_weight}
w_{n}(x)\propto\frac{\mu_{n}(x)}{\mu_{n-1}(x)}.
\end{align}

\subsection{An SMC algorithm for MMLE}

We exploit the connection between mirror descent and SMC samplers established in \cite{chopin2023connection} to derive an SMC algorithm which approximates the iterates~\eqref{eq:md_update_mmle}.
First, we focus on obtaining an approximation of the $\theta$-component of~\eqref{eq:md_update_mmle}.
Assume that we have available a cloud of $N$ weighted particles $\{X_n^i, W_n^i\}_{i=1}^N$ approximating $\mu_n$; in this case we can approximate the $\theta$ update with
\begin{align}
    \label{eq:theta_update}
    \theta_{n+1}^N &= (\nabla h)^{-1}\left(\nabla h(\theta_{n}^N) - \gamma_{n+1}\sum_{i=1}^N W_n^i \nabla_\theta U(\theta_n^N, X_n^i)\right).
\end{align}
If $h=\norm{\cdot}^2/2$, updates of the form~\eqref{eq:theta_update} are popular in the particle filtering literature \citep{poyiadjis2011particle}.

Under the assumption that the sequence $(\theta_n)_{n\geq 0}$ is fixed, the sequence 
$(\mu_n)_{n\geq 0}$ in~\eqref{eq:md_update_mmle} can be approximated through the SMC sampler described above.
The weights $w_n$ are
\begin{align*}
    w_n(x) & \propto \frac{\mu_n(x)}{\mu_{n-1}(x)}= \frac{\mu_{n-1}(x)^{(1-\gamma_{n})}p_{\theta_{n-1}}(x, y)^{\gamma_{n}}}{\mu_{n-1}(x)}=\left(\frac{p_{\theta_{n-1}}(x, y)}{\mu_{n-1}(x)}\right)^{\gamma_{n}}.
\end{align*}

Proceeding recursively, one finds that
\begin{align}
\label{eq:smc_bad}
    \mu_{n+1}(x) &\propto \mu_n(x)^{(1-\gamma_{n+1})}p_{\theta_n}(x, y)^{\gamma_{n+1}} \\
    &\propto \mu_0(x)^{\prod_{k=1}^{n+1}(1-\gamma_k)}p_{\theta_n}(x, y)^{\gamma_{n+1}}p_{\theta_{n-1}}(x, y)^{\gamma_{n}(1-\gamma_{n+1})}\dots p_{\theta_{0}}(x, y)^{\gamma_{1}\prod_{k=2}^{n+1}(1-\gamma_k)}\notag,
\end{align}
which gives the following expression for the weights
\begin{align}
\label{eq:md_weights}
    w_n(x; \theta_{0:n-1}) &= \left(\frac{p_{\theta_{n-1}}(x, y)}{\mu_0(x)^{\prod_{k=1}^{n-1}(1-\gamma_k)}p_{\theta_{n-2}}(x, y)^{\gamma_{n-1}}\dots p_{\theta_{0}}(x, y)^{\gamma_{1}\prod_{k=2}^{n-1}(1-\gamma_k)}}\right)^{\gamma_{n}},
\end{align}
with $\prod_{k=p}^q \cdot= 1$ if $p>q$, which can be computed in $\mathcal{O}(1)$ time in the number of particles.

An SMC approximation of the $\mu$-update in~\eqref{eq:md_update_mmle} is given by a weighted particle population $\{X_n^i, W_n^i\}_{i=1}^N$ approximating $\mu_n$ for each $n$. 
Combining this approximation with the $\theta$-update in~\eqref{eq:theta_update} we obtain the mirror descent algorithm for LVMs (MD-LVM) in Algorithm~\ref{alg:smc_lvm}.

\begin{algorithm}[th]
\begin{algorithmic}[1]
\STATE{\textit{Inputs:} sequence of step sizes $(\gamma_n)_{n\geq 1}$, Markov kernels $(M_n)_{n\geq 1}$, initial proposal $(\theta_0^N, \mu_0)$.}
\STATE{\textit{Initialise:} sample $X_0^i=\widetilde{X}_0^i\sim \mu_0$ and set $W_0^i=1/N$ for $i=1,\dots, N$.}
\FOR{$n\geq 1$}
\STATE{\textit{Update}: set $\theta_{n}^N = (\nabla h)^{-1}\left(\nabla h(\theta_{n-1}^N) - \gamma_{n}\sum_{i=1}^N W_{n-1}^i \nabla_\theta U(\theta_{n-1}^N, X_{n-1}^i)\right)$} 
\IF{$n>1$}
\STATE{\textit{Resample:} draw $ \{\widetilde{X}_{n-1}^i\}_{i=1}^N$ independently from $\{X_{n-1}^i, W_{n-1}^i\}_{i=1}^N$ and set $W_n^i =1/N$ for $i=1,\dots, N$.}
\ENDIF
\STATE{\textit{Propose:} draw $X_n^i\sim M_{n}( \widetilde{X}_{n-1}^i, \cdot;\theta_{0:n-2}^N)$ for $i=1,\dots, N$.}
\STATE{\textit{Reweight:} compute and normalise the weights $W_n^i \propto w_{n}(X_n^i;\theta^N_{0:n-1})$ in~\eqref{eq:md_weights} for $i=1,\dots, N$.}
\ENDFOR
\STATE{\textit{Output:} $(\theta_n^N, \{X_n^i,W_n^i\}_{i=1}^N)$}
\end{algorithmic}
\caption{Mirror descent for latent variable models (MD-LVM).}\label{alg:smc_lvm}
\end{algorithm}

\subsubsection{A practical algorithm for MMLE (SMCs-LVM)}
\label{sec:smc_lvm}
The SMC algorithm described above approximates the iterates~\eqref{eq:md_update_mmle}, However, its complexity increases linearly with $n$ as the weights~\eqref{eq:md_weights} involve an increasing number of terms and so do the target distributions $\mu_n$.
This makes the corresponding SMC scheme impractical if $n$ is large, as it will be the case in high-dimensional problems in which the learning rate $\gamma_n$ needs to be sufficiently small to avoid instabilities in the $\theta$-update.

To alleviate this issue we propose to swap the $\mu$-update in the iteration~\eqref{eq:md_update_mmle} with the tempering one~\eqref{eq:tempering} with $\pi=p_{\theta_n}$, yielding
\begin{align}
\label{eq:md_fast}
    \widetilde{\mu}_{n+1}(x)\propto \mu_0(x)^{\prod_{k=1}^{n+1}(1-\gamma_k)}p_{\theta_n}(x, y)^{1-\prod_{k=1}^{n+1}(1-\gamma_k)}.
\end{align}
The two iterations coincide for fixed $\pi$, but since in our case $p_\theta$ varies at each iteration, \eqref{eq:md_update_mmle} and~\eqref{eq:tempering} are not the same in general.

When using $\widetilde{\mu_n}$-invariant Markov kernels, the importance weights are given by
\begin{align}
\label{eq:md_weights_fast}
    \widetilde{w}_n(x; \theta_{n-2:n-1}) &\propto \frac{\widetilde{\mu}_n(x)}{\widetilde{\mu}_{n-1}(x)}= \frac{p_{\theta_{n-1}}(x, y)^{1-\prod_{k=1}^{n}(1-\gamma_k)}}{p_{\theta_{n-2}}(x, y)^{1-\prod_{k=1}^{n-1}(1-\gamma_k)}\mu_0(x)^{\gamma_{n}\prod_{k=1}^{n-1}(1-\gamma_k)}}.
\end{align}
The weights~\eqref{eq:md_weights_fast} only require $\theta_{n-1}, \theta_{n-2}$ to be computed and therefore have constant complexity in $n$; similarly, one can select the Markov kernels $\widetilde{M}_n$ 
to only depend on $\theta_{n-2}$. We name this algorithm sequential Monte Carlo sampler for LVMs (SMCs-LVM).

To motivate replacing $\mu_n$ with $\widetilde\mu_n$, observe that for $n=1$, $\mu_1\equiv \widetilde\mu_1$, and for all $n\geq 2$ (see Appendix~\ref{app:ratio} for a proof)
\begin{align}
\label{eq:ratio}
    \frac{\mu_{n}(x)}{\widetilde{\mu}_{n}(x)}&\propto \prod_{k=0}^{n-2}\left(\frac{p_{\theta_k}(x, y)}{p_{\theta_{n-1}}(x, y)}\right)^{\gamma_{k+1}\prod_{j=k+2}^n(1-\gamma_j)}\\
    &= \left(\frac{p_{\theta_{0}}(x, y)}{p_{\theta_{n-1}}(x, y)}\right)^{\gamma_1\prod_{j=2}^n(1-\gamma_j)}\dots \left(\frac{p_{\theta_{n-2}}(x, y)}{p_{\theta_{n-1}}(x, y)}\right)^{\gamma_{n-1}(1-\gamma_n)}.\notag
\end{align}
Under our smoothness assumptions, for small step-sizes $(\gamma_n)_{n\geq 1}$, we have $\theta_{n-1}\approx\theta_{n-2}$
and similarly for $\theta_{n-1}\approx\theta_{n-3}$. It follows that the last terms in~\eqref{eq:ratio} are close to 1. For the first terms in~\eqref{eq:ratio}, the discrepancy between $\theta_{n-1}$ and $\theta_k$ is large, but $\gamma_{k+1}\prod_{j=k+2}^n(1-\gamma_j)\approx 0$ for large $n$ since $\gamma_n\leq 1$. It follows that also the initial terms in~\eqref{eq:ratio} are close to 1. This intuition is empirically confirmed by Example~\ref{ex:toy} and the results in Appendix~\ref{app:ratio}.
\begin{example}
Consider the toy LVM given by $x|\theta  \sim  \mathcal{N}(\cdot;\theta\textsf{1}_{d_x}, \textsf{Id}_{d_x})$ and $ y|x \sim \mathcal{N}(\cdot;x, \textsf{Id}_{d_x})$ for $\theta = 1$, $d_x=50$ and one data point $y$.
We can explicitly compute the maximum likelihood estimator $\theta^\star = {d_x}^{-1}\sum_{i=1}^{d_x} y_i$ and the posterior distribution $p_{\theta}(x|y) = \mathcal{N}(x; (y+\theta)/2, 1/2\textsf{Id}_{d_x})$. 
Assumption~\ref{ass:convex} is satisfied with $h = \norm{\cdot}^2/2$ (see Appendix~\ref{app:ratio}).
Replacing~\eqref{eq:md_update_mmle} with~\eqref{eq:md_fast} leads to the same results (Figure~\ref{fig:toy}) but the savings in terms of computational cost are considerable: using~\eqref{eq:md_fast} instead of~\eqref{eq:smc_bad} is about 100 times faster.

\begin{figure}
	\centering
	\begin{tikzpicture}[every node/.append style={font=\normalsize}]
 \node (img1) {\includegraphics[width = 0.4\textwidth]{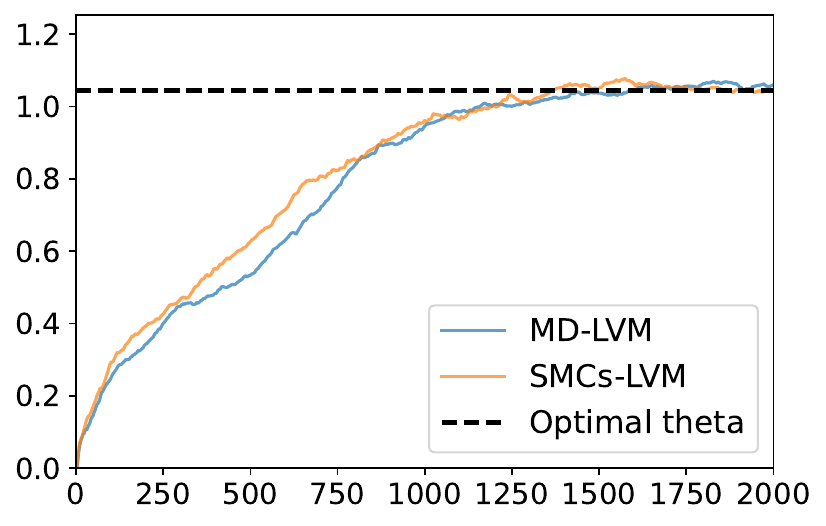}};
  \node[left=of img1, node distance = 0, rotate = 90, anchor = center, yshift = -0.8cm] {$\theta_n$};
  \node[right=of img1, node distance = 0, xshift = -0.5cm] (img2) {\includegraphics[width = 0.4\textwidth]{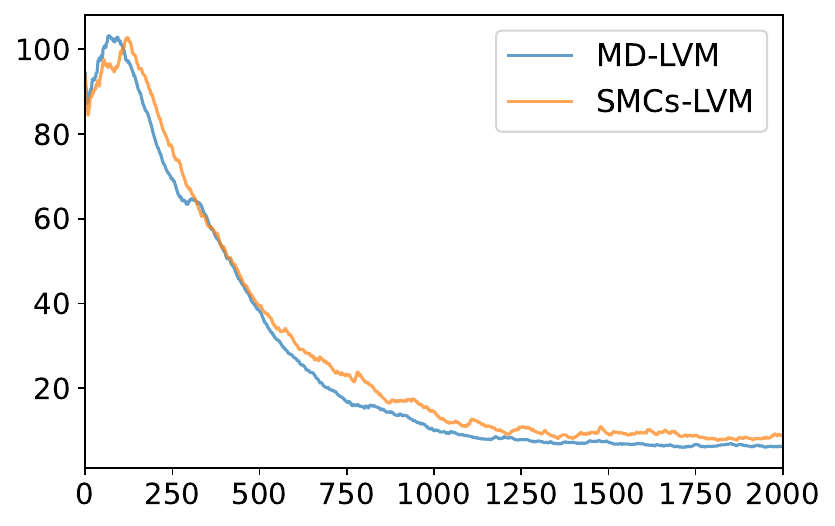}};
	\node[below=of img1, node distance = 0, yshift = 1.2cm] {$n$};
    \node[below=of img2, node distance = 0, yshift = 1.2cm] {$n$};
    \node[left=of img2, node distance = 0, rotate = 90, anchor = center, yshift = -0.8cm] {$\KL(\mu_n|p_{\theta_\star}(\cdot|y))$};
	\end{tikzpicture}
	\caption{Comparison of~\eqref{eq:smc_bad} and~\eqref{eq:md_fast} on a toy Gaussian model. Left: evolution of $\theta$-iterates. Right: evolution of KL divergence from the true posterior $p_{\theta_\star}(x|y)$. }
	\label{fig:toy}
\end{figure}

\end{example}
\subsubsection{Algorithmic setup}

MD-LVM and SMCs-LVM require the specification of the number of iterations $T$, the step sizes $(\gamma_n)_{n\geq 1}$, the initial proposal $\mu_0$, the Markov kernels $M_n, \widetilde{M}_n$, and of $h$.

For our experiments, the initial distribution $\mu_0$ is a standard Gaussian when $\calX\subseteq \real^{d_x}$ and the uniform distribution when $\calX$ is a discrete space. The value of $\theta_0$ is specified for each experiment. 
For the choice of $\mu_{n-1}$-invariant Markov kernels we refer to the wide literature on MCMC (see, e.g., \citet[Chapter 15]{chopin2020introduction}). A common choice, and the one we use in our experiments when  $\calX\subseteq \real^{d_x}$, is a random walk Metropolis kernel whose covariance can be tuned using the particle system \citep{dau2022waste}.

When $\theta\in \real^{d_\theta}$, if $\nabla_\theta U$ is Lipschitz continuous and convex , the natural choice is $h=\norm{\cdot}^2/2$ (i.e. gradient descent). When the domain of $\theta$ is a proper subset of $\real^{d_\theta}$, $h$ could be chosen to enforce the constraints  (e.g. $h(u) = \sum_{i=1}^{d_\theta}u_i\log u_i$ for the $\real^{d_\theta}$-simplex), see Section~\ref{ex:sbm} for an example. For non-Lipschitz $\nabla_\theta U$, $h$ should be selected so that Assumption~\ref{ass:convex} holds (see, e.g., \cite{lu2018relatively}).

The choice of step size $(\gamma_n)_{n\ge 1}$ and number of iterations carries the same difficulties as those encountered in selecting the step size for gradient descent (for the $\theta$-component) and for selecting the appropriate tempering schedule in SMC for the $\mu$-component. While we expect adaptive strategies similar to those employed in the SMC literature \citep{jasra2011inference} to carry over to this context, their use is less straightforward as the target distribution changes at each iteration.

When choosing $(\gamma_n)_{n\ge 1}$, one needs to take into account not only its use in the $\theta$-update~\eqref{eq:theta_update}, for which guidelines on choosing $(\gamma_n)_{n\ge 1}$ are given by Proposition~\ref{cor:convergence} and the vast literature on mirror descent \citep{lu2018relatively}, but also its use in the $\mu$-update.
In particular, the weights~\eqref{eq:md_weights_fast} can be unstable if $\theta_{n-2}, \theta_{n-1}$ are too far apart. While this is partially mitigated by an appropriate choice of $h$ (e.g. for constrained spaces), 
a pragmatic choice is to choose $\gamma_n$ small enough to guarantee that the update for $\theta$ and the mirror descent step for $\mu$ are stable, and then adapt the number of iterations $T$ so that the $\theta$-component has converged.  

\subsection{Convergence properties}

Algorithm~\ref{alg:smc_lvm} provides an approximation of the mirror descent iterates~\eqref{eq:md_update_mmle}. To assess the quality of these approximations we adapt results from the SMC literature (e.g., \cite{smc:theory:Del04}) to our context and combine them with additional results needed to control the effect of the varying $\theta$. 
Since the true sequence $(\theta_n)_{n\geq 0}$ is not known but approximated via~\eqref{eq:theta_update}, Algorithm~\ref{alg:smc_lvm} uses weight functions and Markov kernels which are random and approximate the true but unknown weight functions and kernels. In the case of a fixed $\theta$-sequence, the Markov kernels and weights coincide with those of a standard SMC algorithm and the results for standard SMC samplers (e.g. \cite{smc:theory:Del04}) hold.

We provide convergence bounds for both the $\mu$-iterates and $\theta$-iterates. 
As in standard SMC literature, in the case of the $\mu$-iterates we focus on the approximation error for measurable bounded test functions $\testfn:\calX\to\real$ with $\supnorm{\testfn}:= \sup_{x\in  \calX}\vert \testfn(x)\vert<\infty$, a set we denote by $\bounded$. 
We make the following stability assumptions on $M_n$ and $w_n$:
\begin{assumption}
\label{ass:kernel}
Let the dependence of $M_{n}$ on $\theta$ be explicit and define $M_{n, \theta_{0:n-2}}(x, \cdot):=M_n(x, \cdot;\theta_{0:n-2})$.
The Markov kernels $M_{n, \theta_{0:n-2}}$ are stable with respect to $(\theta_n)_{n\geq 0}$, that is, there exists a constant $\rho>0$ such that for all measurable bounded functions $\testfn\in \bounded$
\begin{align*}
    \vert M_{n,\theta_{0:n-2}}\testfn(x) - M_{n,\theta_{0:n-2}'}\testfn(x) \vert\leq \rho\supnorm{\testfn}\sum_{j=0}^{n-2}\norm{\theta_j-\theta_j'}
\end{align*}
for all $(\theta_{0:n-2}, \theta_{0:n-2}')\in(\real^{d_\theta})^{n-1}$ uniformly in $x\in\calX$, where we denote $M\testfn(x) = \int
M(x, dv) \testfn(v)$ for all $x\in\calX$, $\testfn\in \bounded$ and any Markov kernel $M$.
\end{assumption}

\begin{assumption}
\label{ass:weights}
The weights~\eqref{eq:md_weights} are bounded above, i.e. $\supnorm{w_n}<\infty$, and $\supnorm{\nabla_\theta U}<\infty$.
\end{assumption}

Assumption~\ref{ass:kernel} is a technical assumption which ensures that the kernels $M_n$ are well-behaved.
While strong, this assumption has been previously considered in the SMC literature to deal with adaptive kernels \citep{beskos2016convergence}.
The stability conditions in \cite{caprio2023calculus} relate expressions similar to that in Assumption~\ref{ass:kernel} to the invariant measure of the corresponding kernels. In our case, this would translate into a stability with respect to $\theta$ of the joint likelihood $p_\theta(x, y)$.

Assumption~\ref{ass:weights} requires the weights to be bounded above, a standard assumption in the SMC literature.
For simplicity, we also assume that $\supnorm{\nabla_\theta U}<\infty$, which implies that the weights are stable as shown in Appendix~\ref{app:proof}.
While this assumption is often not satisfied
, we point out that the results we obtain hold under weaker integrability assumptions (see, e.g. \cite{agapiou2017importance}) by further assuming that the weights $w_n$ are Lipschitz continuous in the sequence $(\theta_n)_{n\geq 0}$, at the cost of significantly complicating the analysis.

The following convergence result for Algorithm~\ref{alg:smc_lvm} is proven in Appendix~\ref{app:proof}.
\begin{prop}
\label{prop:lp} Under Assumption~\ref{ass:convex}--\ref{ass:weights}, if $h=\norm{\cdot}^2/2$
, for every time $n\geq 0$, every $p\geq 1$, $N\geq 1$ and $(\gamma_n)_{n\geq 1}$ such that $1\geq\gamma_n\geq \gamma_{n-1}\geq 0$ there exist finite non-negative constants $C_{p, n}, D_{p,n}$ such that for every measurable bounded function $\testfn\in \bounded$
\begin{align*}
\Exp\left[\left\lvert\sum_{i=1}^NW_n^i \testfn(X_n^i) -\int \testfn(x)\mu_n(x)dx\right\rvert^p\right]^{1/p} &\leq C_{p,n}\frac{\supnorm{\testfn}}{N^{1/2}},\\
\Exp\left[\Vert \theta_{n}^N-\theta_{n}\Vert^p\right]^{1/p} &\leq D_{p, n}\frac{\gamma_{n} }{N^{1/2}}.
\end{align*}
\end{prop}

The first result in Proposition~\ref{prop:lp} quantifies the maximum approximation error for the $\mu$-iterates, while the second one quantifies the maximum error in recovering the $\theta$-iterates. An equivalent convergence result can be obtained for the algorithm described in Section~\ref{sec:smc_lvm} and the iterates~\eqref{eq:md_fast}.
If $\calX= \real^{d_x}$, under 
additional assumptions on the regularity of $U$
, we can obtain the global error achieved by Algorithm~\ref{alg:smc_lvm}. The proof is given in Appendix~\ref{app:proof}.

\begin{cor}
\label{cor:parameter}
Assume that $\theta \mapsto p_\theta(\cdot|y)$ is twice differentiable and that $p_\theta(x, y)>0$ for all $(\theta, x)\in \real^{d_\theta}\times \real^{d_x}$. Under Assumption~\ref{ass:convex}--\ref{ass:weights} with $l, L>0$, if $h=\norm{\cdot}^2/2$
, we have
\begin{align*}
    \Exp[\norm{\theta_n^N-\theta^\star}^2 ]^{1/2}\leq \sqrt{\frac{2}{l}\frac{\KL(p_{\theta^\star}(\cdot|y)|\mu_0)+\norm{\theta^\star-\theta_0}^2}{\gamma_1} \prod_{k=1}^n(1-\gamma_k\min(l, 1))}+D_{2, n}\frac{\gamma_n}{N^{1/2}},
\end{align*}
where $D_{2, n}$ is given in Proposition~\ref{prop:lp}.
\end{cor}

In the case $\gamma_n\equiv \gamma$, Corollary~\ref{cor:parameter} gives $\Exp[\norm{\theta_n^N-\theta^\star}^2 ]^{1/2}=\mathcal{O}\left((1-\gamma)^{n/2}+\gamma N^{-1/2}\right)$. The first term decays exponentially fast in the number of iterations and accounts for the convergence of mirror descent to $(\theta^\star, p_{\theta^\star})$, while the second term corresponds to the Monte Carlo error and discretisation bias.
Comparing this result with \citet[Theorem 3.8]{akyildiz2023interacting} and \citet[Eq. (9)]{caprio2024error} we find that our algorithm achieves an equivalent convergence rate in terms of the key parameters $\gamma, N, n$.


\section{Numerical experiments}
\label{sec:ex}

We compare our methods with popular alternatives in the literature: the particle gradient descent algorithm (PGD; \cite{kuntz2023particle}) and the interacting particle Langevin algorithm (IPLA; \cite{akyildiz2023interacting}) when $\X=\real^{d_x}$ and sequential Monte Carlo for marginal maximum likelihood (SMC-MML; \cite{johansen2008particle}) and variants of expectation maximization (EM) when PGD and IPLA cannot be applied.
All experimental details and additional results are available in the Supplement.
\subsection{Toy examples}
\label{ex:toy}

\paragraph{Gaussian Mixture}

\begin{figure}
	\centering
	\begin{tikzpicture}[every node/.append style={font=\normalsize}]
		\node (img1) {\includegraphics[width = 0.4\textwidth]{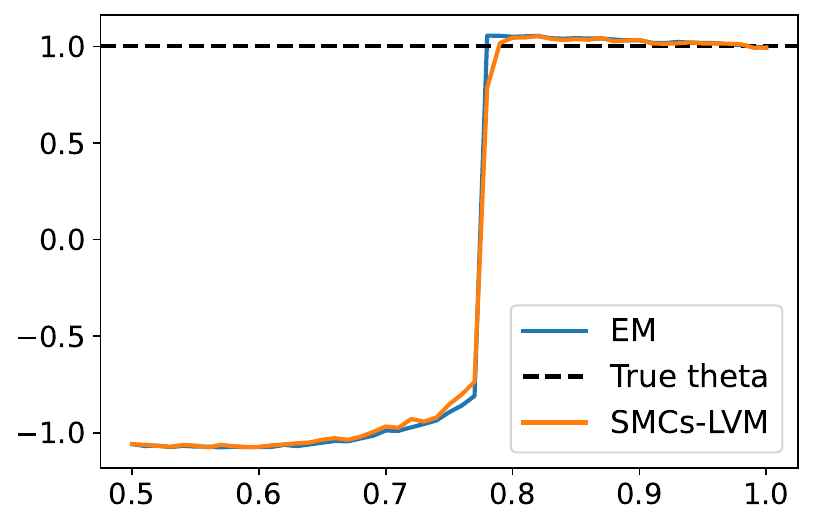}};
		\node[right=of img1, node distance = 0, xshift = -0.5cm] (img2) {\includegraphics[width = 0.4\textwidth]{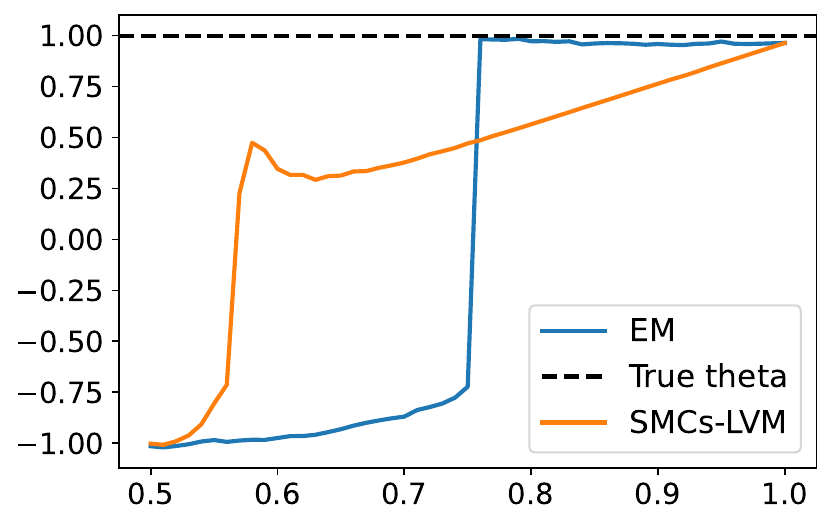}};
  \node[left=of img1, node distance = 0, rotate = 90, anchor = center, yshift = -0.8cm] {$\theta_T$};
    \node[left=of img2, node distance = 0, rotate = 90, anchor = center, yshift = -0.8cm] {$\theta_T$};
		\node[below=of img1, node distance = 0, yshift = 1.2cm] {$\alpha$};
		\node[below=of img2, node distance = 0, yshift = 1.2cm] {$\alpha$};
	\end{tikzpicture}
    \vskip -10pt
	\caption{Comparison of final $\theta$-iterate for different $\alpha$s for EM and SMCs-LVM for the Gaussian mixture model. Left: SMCs-LVM with uniform proposal. Right: SMCs-LVM with proposal $p(x)$.}
	\label{fig:em}
\end{figure}
Take the symmetric Gaussian mixture $
    p_\theta(y) = \alpha \mathcal{N}(y;\theta, 1)+(1-\alpha)\mathcal{N}(y;-\theta, 1)$,
where the latent variable $x$ corresponds to the allocation of each observation to one of the mixture components (see Appendix~\ref{app:mix} for the model specification and experimental setup).  
We simulate $1000$ data points from the model with $\theta=1$, consider $\alpha\in [0.5, 1]$ and select $\theta_0=-2$.
In the case of known $\alpha$, \citet[Theorem 1]{xu2018benefits} shows that for values of $\alpha$ close to 0.5 and starting point $\theta_0\leq -\theta$ the standard EM algorithm converges to a local maximum.

We compare EM with SMCs-LVM with two choices of $\widetilde{M}_n$: a random walk Metropolis with uniform proposal (Figure~\ref{fig:em} left) and a random walk Metropolis with proposal $p_\theta(x)=p(x)$ (Figure~\ref{fig:em} right).
EM struggles to converge to the global maximum for $\alpha\approx 0.5$, although it performs well for $\alpha>0.75$.
When using a uniform proposal the results of SMCs-LVM coincide with those of EM. When using $p(x)$ (which carries information on $\alpha$) SMCs-LVM outperforms EM for $0.5\leq\alpha<0.75$ but is less accurate for larger $\alpha$. EM generally converges faster than SMCs-LVM, but, when using the same number of iterations the runtime of SMCs-LVM is less than twice that of EM. 


\paragraph{Multimodal marginal likelihood}
To show the limitations of PGD and IPLA we consider a 
popular benchmark for multimodality from \cite{gaetan2003multiple}.
The model is $p_\theta(x) = \textrm{Gamma}(x;\alpha, \beta)$, $p_\theta(y|x) = \N(y;\theta, x^{-1})$
with $\alpha = 0.525, \beta = 0.025$. The observed data are $\{-20, 1, 2, 3\}$; 
$p_\theta(y)$ is multimodal and a global maximum is located at $1.997$.


For this example, $\nabla_\theta U(\theta, x)$ satisfies Assumption~\ref{ass:convex} only locally (see Appendix~\ref{app:multimodal}).
In addition, $\nabla_x U(\theta, x)$ is not Lipschitz continuous in $x$, this causes the ULA update employed in PGD and IPLA to be unstable as shown in Figure~\ref{fig:multimodal}.
As a consequence, PGD and IPLA fail to recover the MLE and the corresponding posterior and
the $\theta$-iterates converge to a value which is far from the global optimum.

\begin{figure}
	\centering
	\begin{tikzpicture}[every node/.append style={font=\normalsize}]
		\node(img4) {\includegraphics[width = 0.3\textwidth]{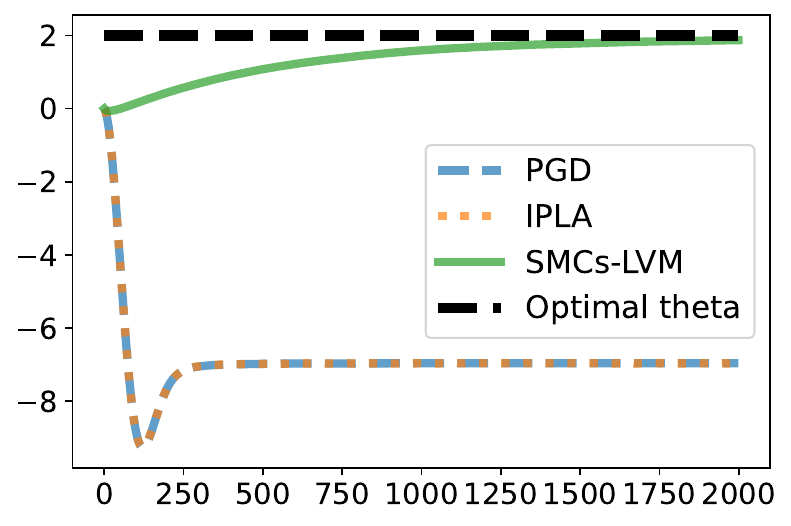}};
		\node[right=of img4, node distance = 0, xshift = -0.5cm] (img5) {\includegraphics[width = 0.3\textwidth]{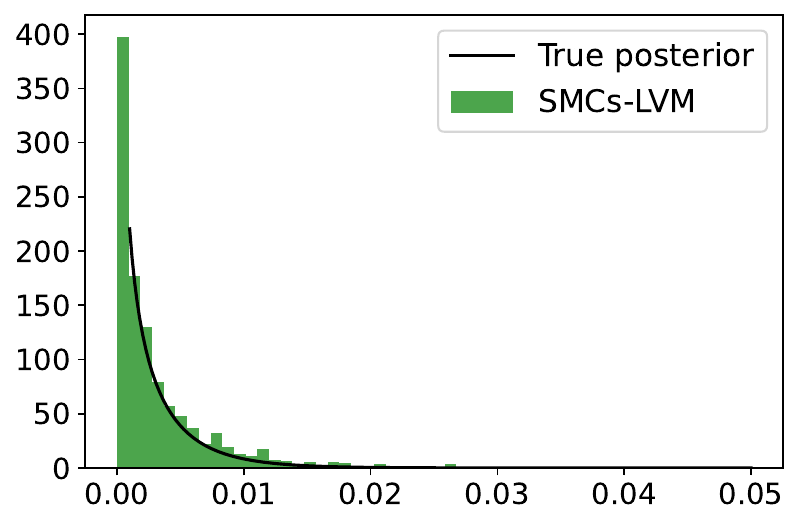}};
  \node[right=of img5, node distance = 0, xshift = -0.5cm] (img6) {\includegraphics[width = 0.3\textwidth]{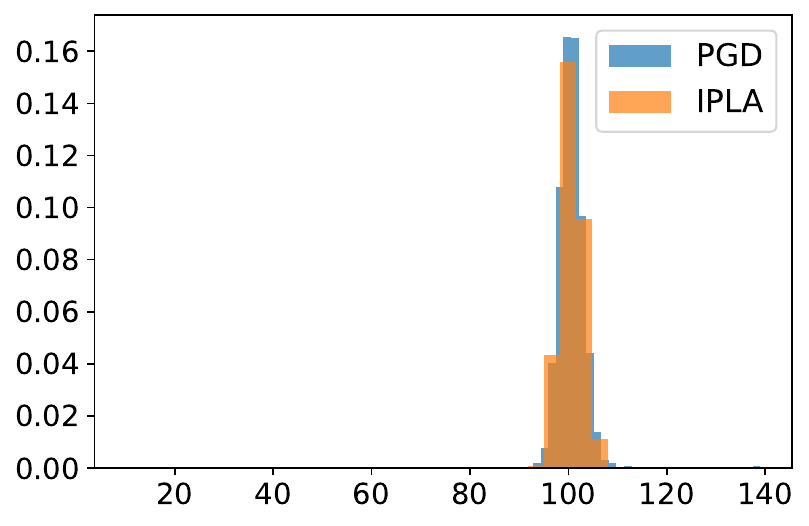}};
		\node[below=of img4, node distance = 0, yshift = 1.2cm] {$n$};
		\node[below=of img5, node distance = 0, yshift = 1.2cm] {$x$};
  		\node[below=of img6, node distance = 0, yshift = 1.2cm] {$x$};
        \node[left=of img4, node distance = 0, rotate = 90, anchor = center, yshift = -0.8cm] {$\theta_n^N$};
        \node[left=of img5, node distance = 0, rotate = 90, anchor = center, yshift = -0.8cm] {$p_{\theta_T^N}(x|y)$};
        \node[left=of img6, node distance = 0, rotate = 90, anchor = center, yshift = -0.8cm] {$p_{\theta_T^N}(x|y)$};
	\end{tikzpicture}
    \vskip -20pt
	\caption{Comparison of  PGD, IPLA and SMCs-LVM for the multimodal marginal likelihood example. Left: evolution of $\theta$-iterates. Middle: final posterior approximation for SMCs-LVM. Right: final posterior approximation for PGD and IPLA.}
	\label{fig:multimodal}
\end{figure}

We further compare SMCs-LVM with SMC-MML. Both methods are based on SMC and provide approximations of the posterior but SMC-MML uses a cloud of particles to approximate the MLE while our method uses a gradient step to converge to the MLE. We set $N=100, T=50$ for both algorithms, and $\gamma_n\equiv 0.05$ for SMCs-LVM so that convergence occurs in the same number of iterations.
SMCs-LVM is approximately 15 times faster than SMC-MML, and its mean squared error (over 100 replicates) is twice that of SMC-MML. This is likely due to the fact that SMC-MML averages over $N$ particles to obtain the estimate of $\theta^\star$, while SMCs-LVM uses only one sample.

\subsection{Bayesian logistic regression}
\label{ex:blr}

We compare SMCs-LVM with PGD and IPLA on a simple Bayesian regression task which satisfies Assumption~\ref{ass:convex} with $h=\norm{\cdot}^2/2$ (see Appendix~\ref{app:blr}) and 
for which both PGD and IPLA are stable.
The model is $p_\theta(x) = \mathcal{N}(x;\theta ,  \textsf{Id}_{d_x})$, $p_\theta(y|x) = \prod_{j=1}^{d_y}s(v_j^Tx)^{y_j}(1-s(v_j^Tx))^{1-y_j}$,
where $s(u):=e^u/(1+e^u)$ is the logistic function.
For this example, we set $d_x = d_\theta = 3$ and $\theta = (2, 3, 4)$. We simulate 900 data points as follows: we simulate synthetic $d_x$-dimensional covariates $v_j \sim \textsf{Unif}(-1, 1)^{\otimes d_x}$ for $j=1, \dots, 900$ and synthetic data $\{y_j\}_{j=1}^{900}$ from a Bernoulli random variable with parameter $s(v_j^Tx)$.

The update for $\theta$ is identical for PGD and SMCs-LVM while IPLA has an additional noise term; the update for $\mu$ is based on the unadjusted Langevin algorithm for PGD and IPLA, SMCs-LVM employs SMC.
We set $\gamma_n\equiv 0.001$ and $T=2000$, $\theta_0=(0, 0, 0)$, and $X_0$ is sampled from $\mathcal{N}(0, \textsf{Id})$.

We compute the variance of the MLE and its computational cost over 100 repetitions of each method (Table~\ref{tab:bayesian_lr}).
PGD and SMCs-LVM have similar accuracy, but the cost of the latter is about 8 times higher. In fact, the update~\eqref{eq:md_fast} requires evaluating the weights and performing one MCMC step while PGD and IPLA only perform one MCMC step.
IPLA returns estimates with higher variance because of the presence of the noise term in the $\theta$-update which would require smaller $\gamma$ to be reduced (see Figure~\ref{fig:bayesian_lr_comparison} in Appendix~\ref{app:blr}).
\begin{table}
\centering
\begin{tabular}{l|cc|cc|cc}
 & \multicolumn{2}{c}{$N=10$} & \multicolumn{2}{c}{$N = 50$}& \multicolumn{2}{c}{$N = 100$}\\
\hline\noalign{\smallskip}
Method & variance & runtime (s) & variance & runtime (s) & variance & runtime (s)\\
\hline\noalign{\smallskip}
PGD & $8.04\cdot10^{-5}$ &\textbf{0.78}& $2.01\cdot10^{-5}$ & 2.85& $6.40\cdot10^{-6}$& \textbf{7.48}\\
IPLA &1.08 &0.80&$1.69\cdot10^{-1}$&\textbf{2.70}&$8.48\cdot10^{-2}$& 7.57\\
SMCs-LVM &$\mathbf{1.90\cdot10^{-5}}$ &4.57&$\mathbf{3.54\cdot10^{-6}}$& 25.76&$\mathbf{1.98\cdot10^{-6}}$& 50.69 \\
\end{tabular}
\caption{Variance of estimates of the first component of $\theta$ for the Bayesian logistic regression model with $N=10, 50, 100$ and their computational times. $\gamma = 0.001, T=6000$ throughout all experiments.  The best values are in bold.  The behaviour for the remaining two components is equivalent and reported in Appendix~\ref{app:blr}.}
\label{tab:bayesian_lr}
\end{table}

\subsection{Stochastic block model}
\label{ex:sbm}
A stochastic block model (SBM) is a random graph model in which the presence of an edge is determined by the two latent variables associated with the nodes the edge connects, which indicate membership to a block. 
Given an undirected graph with $d_x$ nodes we describe the data generating process as follows: to each node is assigned a latent variable $x$ with categorical distribution $p_\theta(x) = \mathbb{P}(x=q) =p_q$ for $q=1, \dots, Q$, where $Q$ denotes the number of blocks.
Given two nodes $i, j$ the probability of observing an edge $y_{ij}$ connecting them depends on the block membership of $i, j$ and is given by $y_{ij}|x_i, x_j \sim \textrm{Bernoulli}(\nu_{x_ix_j})$,
so that $p_\theta(y|x) =\prod_{i, j=1}^{d_x} (1-\nu_{x_ix_j})^{1-y_{ij}}\nu_{x_ix_j}^{y_{ij}}$. 
The set of parameters 
is $\theta = \left((p_q)_{q=1}^Q, (\nu_{ql})_{q,l=1}^Q)\right)$.

As the latent variables are discrete, IPLA and PGD cannot be applied.
We compare SMCs-LVM with Stochastic Approximation EM (SAEM) through the mean squared error (MSE) for $\theta$ and the Adjusted Rand Index (ARI; \cite{hubert1985comparing}), which compares the posterior clustering of the nodes to the true block memberships
. Higher values of the ARI indicate better recovery of the latent block's membership.

The parameters of this model are all probabilities, to enforce this constraint we use the component-wise logarithmic barrier $h(t) = -\log(t-t^2)$. Assumption~\ref{ass:convex} is not satisfied for this $h$, but the model is convex (see Appendix~\ref{app:sbm}). We also consider the results obtained when using $h = \norm{\cdot}^2/2$.

We select the learning rate for SAEM to be $\gamma_n = 1/n$, with $n$
denoting the iteration number, which satisfies the conditions in \cite{delyon1999convergence}
to guarantee convergence, with this choice SAEM converges in $T=500$ iterations.
Since SAEM associates to each latent variable $x$ a single Markov chain, we set $N=d_x$ to put SMCs-LVM on equal footing. We initialise all components in $\theta$ at $0.3$ and set $\mu_0$ as well as the  proposal for the MCMC kernels to be uniform over the block memberships.
SMCs-LVM is more sensitive to the choice of $\theta_0$ than SAEM, therefore a pragmatic choice would be to initialise SMCs-LVM at the value of $\theta$ obtained after one iteration of SAEM \citep{polyak1992acceleration}.

\subsubsection{Synthetic dataset}
We consider the setup of \cite{kuhn2020properties} with $Q=2$, $p_1=0.6, p_2= 1-p_1 = 0.4, \nu_{11} = 0.25, \nu_{12}=\nu_{21} = 0.1$ and $\nu_{22}=0.2$, and generate one graph with $d_x = 100$ nodes from the corresponding model. To achieve convergence in $T=500$ iterations we set $\gamma_n\equiv 0.06$ for SMCs-LVM with the logarithmic barrier (LB) and $\gamma_n\equiv 0.01$ when $h=\norm{\cdot}^2/2$ to guarantee that the gradient descent (GD) update is stable.  

SMCs-LVM consistently outperforms SAEM in terms of ARI and MSE for $\nu_{ij}$ (Figure~\ref{fig:sbm}), the gain in ARI is of $30\%$ for LB and $60\%$ for GD. The runtime for SMCs-LVM is 2.5 times that of SAEM. 
The GD update provides more accurate results but requires smaller $\gamma_n$ which results in slower convergence.

\begin{figure}
	\centering
	\begin{tikzpicture}[every node/.append style={font=\normalsize}]
		\node (img1) {\includegraphics[width = 0.4\textwidth]{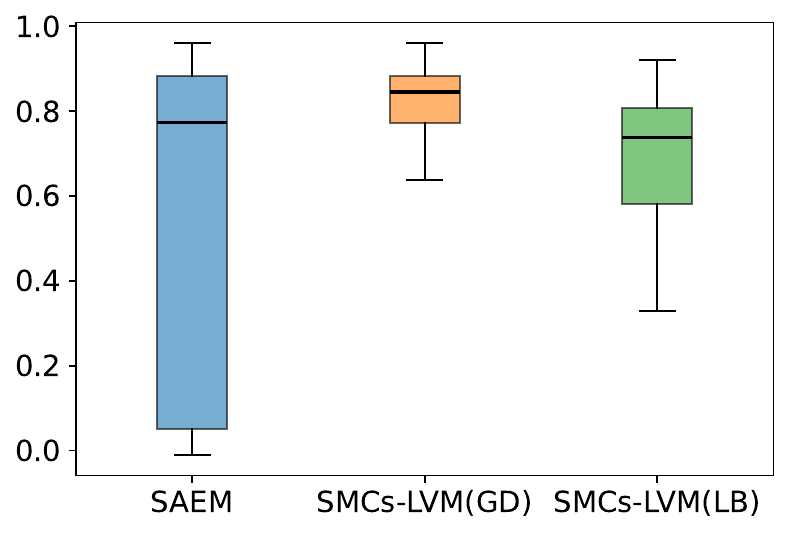}};
		\node[right=of img1, node distance = 0, xshift = -0.5cm] (img2) {\includegraphics[width = 0.4\textwidth]{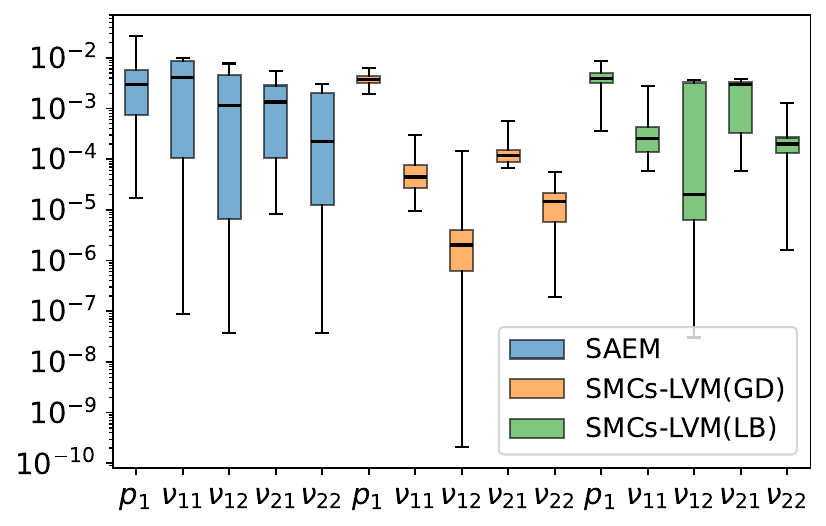}};
  \node[left=of img1, node distance = 0, rotate = 90, anchor = center, yshift = -0.8cm] {ARI};
    \node[left=of img2, node distance = 0, rotate = 90, anchor = center, yshift = -0.8cm] {MSE};
	\end{tikzpicture}
	\caption{Distribution of ARI and MSE for 50 repetitions of SAEM and SMCs-LVM with logarithmic barrier (LB) and gradient descent (GD) update for $\theta$ for the stochastic block model on synthetic data.}
	\label{fig:sbm}
\end{figure}

\subsubsection{Real dataset}

Consider the karate club network with $d_x = 34$ nodes \citep{zachary1977information}. The SBM with 2 blocks is known to separate high-degree nodes from low-degree ones when applied to this network \citep{karrer2011stochastic}. We fit this model with $Q=2$ for 50 times and compare the results obtained with SAEM and SMCs-LVM.
We test the speed of convergence of the two methods, and stop iterating when $\underset{i = 1, \dots 5}{\max} [\theta_n^N(i)-\theta_{n-1}^N(i)]^2< 10^{-7}$, with $\theta_n^N(i)$ denoting the $i$-th component of the parameter vector.
We set $\gamma_n\equiv 0.1$ for SMCs-LVM (both GD and LB). 
SMCs-LVM is about 4 times slower than SAEM but the average ARI is considerably higher (Table~\ref{tab:karate}).

\begin{table}
    \centering
    \begin{tabular}{c|ccc}
Method & $T$ & runtime (s) & ARI\\
    \hline\noalign{\smallskip}
     SAEM    & 99 & 1.06 & 0.77 \\
     SMCs-LVM(LB)    & 161 & 4.66 & 0.99\\
     SMCs-LVM(GD)    & 158 & 4.54 & 0.97
    \end{tabular}
    \caption{Average speed of convergence  and ARI for 50 repetitions of SAEM and SMCs-LVM with logarithmic barrier (LB) and gradient descent (GD) update for $\theta$ for the stochastic block model on the karate club network.}
    \label{tab:karate}
\end{table}
\section{Conclusions}
\label{sec:conclusion}

We introduced a sequential Monte Carlo implementation of a mirror descent approach to perform joint parameter inference and posterior estimation in latent variable models. The algorithm applies to discrete latent variables and only requires uniform relative convexity and smoothness in the $\theta$-component. Our experiments show that the algorithm is effective if these conditions are satisfied locally.

Our work is closely related to \cite{kuntz2023particle, akyildiz2023interacting}, but can be applied to LVMs whose log-likelihood is not differentiable in $x$ without restricting to those that are convex and lower-semicontinuous as required by \cite{encinar2024proximal}.
For LVMs in which the log-likelihood is convex and gradient-Lipschitz in both variables (Section~\ref{ex:blr}), SMCs-LVM is competitive with methods based on Langevin sampling.

When compared to EM and its variants, our approach suffers less from local maxima but, as most gradient methods, is more sensitive to initialisation of the parameters (Section~\ref{ex:toy} and~\ref{ex:sbm}). The posterior approximations provided by SMCs-LVM outperforms that of SAEM. In fact, while SAEM attempts to directly sample from the posterior at every iteration, SMCs-LVM uses a tempering approach, slowly bridging from an easy-to-sample-from distribution to the posterior.

Several extensions of the methods proposed here are possible: mirror descent allows natural extensions to constrained optimisation and non-Lipschitz settings by appropriate choice of the Bregman divergence $B_h$ in Assumption~\ref{ass:convex} \citep{lu2018relatively}. Furthermore, improved accuracy could be achieved by replacing the $\theta$-update by analytic maximisation whenever possible \citep[Appendix D]{kuntz2023particle}, or by considering more terms in~\eqref{eq:md_fast} to achieve a better trade-off between computational cost and accuracy. One option could be to set a fixed lag $L$ and only consider the most recent $L$ iterations, i.e. $p_{\theta_{n-L+1}}(x, y)$ to $p_{\theta_n}(x, y)$. Alternatively, one could discard all terms in~\eqref{eq:smc_bad} for which
$\gamma_k\prod_{j=k+1}^{n+1}(1-\gamma_j)<\varepsilon$ for some $\varepsilon>0$.
Our methods could also be extended by sampling several times from the Markov kernels $\widetilde{M}_n$ and reweighting all the generated samples in the spirit of \cite{dau2022waste}; note that this would not be feasible for SAEM, since for this method no reweighting step is performed. 

In summary, our proposed methods add to the list from which practitioners can choose to infer parameters and posterior in LVMs and can be applied when gradient-based sampling methods cannot, our experiments show that SMCs-LVM outperforms EM and variants in settings in which the posterior is hard to approximate by introducing a tempering approach~\eqref{eq:md_fast}.

\paragraph*{Acknowledgements}
The author wishes to thank Nicolas Chopin and Adam M. Johansen for helpful feedback on a preliminary draft.
\bibliographystyle{apalike}  
\bibliography{biblio_lvm_md.bib}

\clearpage
\appendix 
\section{Ingredients of MD for MMLE}
\label{app:ingredients_md}

\subsection{Derivative and Bregman Divergence}

\begin{proof}[Proof of Proposition~\ref{prop:derivative}]
\begin{enumerate}
\item
    Using the definition of derivative we have, for $z_i = (\theta_i, \mu_i)$, $i=1, 2$, and $\xi = z_2-z_1 = (\xi_\theta, \xi_\mu)$,
    \begin{align*}
        \cF(z_1+\epsilon  \xi) -\cF(z_1) &=\int\log  U(\theta_1+\epsilon \xi_\theta, x)[\mu_1+\epsilon\xi_\mu]( x)dx + \int \log\left( [\mu_1+\epsilon\xi_\mu](x)\right)[\mu_1+\epsilon\xi_\mu]( x)d x\\
        &-\int\log U(\theta_1, x)\mu_1( x)dx -\int \log\left( \mu_1(x)\right)\mu_1( x)d x\\
        &= \int [U(\theta_1+\epsilon \xi_\theta, x) -U(\theta_1, x)]\mu_1( x)dx+ \epsilon\int  U(\theta_1+\epsilon \xi_\theta, x)\xi_\mu( x)dx\\
        &+\int \log\left(1+\epsilon \frac{\xi_\mu(x)}{ \mu_1(x)}\right)\mu_1( x)d x + \epsilon\int \log\left( [\mu_1+\epsilon\xi_\mu](x)\right)\xi_\mu( x)d x.
    \end{align*}
    We then have that 
    \begin{align*}
        \lim_{\epsilon \rightarrow 0}\frac{1}{\epsilon}(\cF(z_1+\epsilon  \xi) -\cF(z_1)) & = \xi_\theta\int \nabla_\theta U(\theta_1, x)\mu_1( x)dx+ \int U(\theta_1, x)\xi_\mu( x)dx\\
        &+\int  \frac{\xi_\mu(x)}{ \mu_1(x)}\mu_1( x)d x + \int \log\left( \mu_1(x)\right)\xi_\mu( x)d x,
    \end{align*}
    where the equality follows from the Taylor expansion of the logarithm as $\epsilon\rightarrow0$.
    Therefore, we can write
    \begin{align*}
        \lim_{\epsilon \rightarrow 0}\frac{1}{\epsilon}(\cF(z_1+\epsilon  \xi) -\cF(z_1)) & = \left\langle\begin{matrix}
         \int \nabla_\theta U(\theta_1, x)\mu_1( x)dx\\
         \log\mu_1(x) + U(\theta_1, x)   +1
        \end{matrix}\ ,\ \begin{matrix}
            \xi_\theta\\
            \xi_\mu
        \end{matrix}
        \right\rangle.
    \end{align*}
    The result follows since $\nabla \cF$ is defined up to additive constants.

\item 
We have
\begin{align*}
    B_\phi(z_1|z_2) &= \KL(\mu_1|\mu_2)+B_h(\theta_1| \theta_2)\\
    &= \int \log (\mu_1(x))\mu_1(x)dx - \int \log (\mu_2(x))\mu_2(x)dx - \ssp{\log \mu_2, \mu_1- \mu_2}\\
    &+h(\theta_1) - h(\theta_2) - \ssp{\nabla h(\theta_2), \theta_1-\theta_2}\\
    &=\phi(\theta_1, \mu_1) - \phi(\theta_2, \mu_2) - \ssp{\nabla \phi(\theta_2, \mu_2), (\theta_1, \mu_1)-(\theta_2, \mu_2)}
\end{align*}
where $\phi(\theta, \mu) = \int \log (\mu(x))\mu(x)dx + h(\theta)$ and $\nabla \phi = (\nabla h, \log \mu)$.
\end{enumerate}
\end{proof}

\subsection{Proof of Proposition~\ref{cor:convergence}}
We first state a preliminary result, 
known as the "three-point inequality" or "Bregman proximal inequality". The result in $\real^d$ can be found in \citetapp[Lemma 3.1]{lu2018relatively} while that over $\cP(\calX)$ in \citetapp[Lemma 3]{aubin2022mirror}.

\begin{lemma}[Three-point inequality]
\label{lem:three-point}
Given $t\in\X$ and some proper convex functional $\cG:\X\rightarrow \real\cup\{+\infty\}$, if $\nabla\phi(t)$ exists, as well as $\bar{z}=\argmin_{z \in \X}\{\cG(z)+B_{\phi}(z | t)\}$, then for all $z\in \X \cap \dom(\phi) \cap \dom(\cG)$: 
\begin{equation*}
\cG(z)+ B_{\phi}(z| t) \ge \cG(\bar{z}) + B_{\phi}(\bar{z}| t) +  B_{\phi}(z| \bar{z}).
	\end{equation*}
\end{lemma}

We can now prove Proposition~\ref{cor:convergence}.
Using the fact that the function $\theta \mapsto U(\theta, x)$ is relatively smooth w.r.t. $B_h$ we have, for all step sizes $\gamma_{n+1}\leq 1/L$
\begin{align}
\label{eq:md_conv1}
    \cF(z_{n+1}) &= \int U(\theta_{n+1}, x)\mu_{n+1}( x)d x + \int \log\left( \mu_{n+1}(x)\right)\mu_{n+1}(x)d x\\
    &\leq \int [U(\theta_n, x)+\ssp{\nabla_\theta    U(\theta_n, x), \theta_{n+1}-\theta_n}+\frac{1}{\gamma_{n+1}} B_h(\theta_{n+1}|\theta_n)]\mu_{n+1}( x)d x \notag\\
    &\qquad+ \int \log\left( \mu_{n+1}(x)\right)\mu_{n+1}(x)d x.\notag
\end{align}
Applying Lemma \ref{lem:three-point} to the convex function $\cG_n(\theta)=\gamma_{n+1} \ssp{\nabla_\theta    U(\theta_n, x), \theta-\theta_n}$, with $t=\theta_n$ and $\bar{z}=\theta_{n+1}$ and Bregman divergence $B_h$ yields
\begin{equation*}
    \ssp{\nabla_\theta    U(\theta_n, x), \theta_{n+1}-\theta_n} + \frac{1}{\gamma_{n+1}} B_h(\theta_{n+1}|\theta_n) \le \ssp{\nabla_\theta    U(\theta_n, x), \theta-\theta_n} + \frac{1}{\gamma_{n+1}} B_h(\theta| \theta_n) - \frac{1}{\gamma_{n+1}} B_h(\theta| \theta_{n+1}).
\end{equation*}
Fix $\theta$, then~\eqref{eq:md_conv1} becomes
\begin{align*}
    \cF(z_{n+1}) &\leq \int [U(\theta_n, x)+\ssp{\nabla_\theta    U(\theta_n, x), \theta-\theta_n} + \frac{1}{\gamma_{n+1}} B_h(\theta| \theta_n) - \frac{1}{\gamma_{n+1}} B_h(\theta| \theta_{n+1})]\mu_{n+1}( x)d x \\
    &\qquad+ \int \log\left( \mu_{n+1}(x)\right)\mu_{n+1}(x)d x.\notag
\end{align*}
Since $\theta \mapsto U(\theta, x)$ is $l$-strongly convex w.r.t. $B_h$, we also have:
\begin{equation*}
    \ssp{\nabla U(\theta_n, x),\theta-\theta_n}\le U(\theta, x)-U(\theta_n, x) - l B_h(\theta|\theta_n), 
\end{equation*}
and the above becomes
\begin{align}
\label{eq:F_decrease1}
    \cF(z_{n+1}) &\leq \int U(\theta, x)\mu_{n+1}( x)d x  + \left(\frac{1}{\gamma_{n+1}} - l\right)B_h(\theta| \theta_n) - \frac{1}{\gamma_{n+1}} B_h(\theta| \theta_{n+1})+ \int \log\left( \mu_{n+1}(x)\right)\mu_{n+1}(x)d x\\
    &\leq \left(\frac{1}{\gamma_{n+1}} - l\right)B_h(\theta| \theta_n) - \frac{1}{\gamma_{n+1}} B_h(\theta| \theta_{n+1})+\KL(\mu_{n+1}|p_\theta).\notag
\end{align}
Since the reverse KL with respect to any target is 1-relatively smooth with respect to the $\KL$ \citepapp{aubin2022mirror, chopin2023connection} we further have, for all $\gamma_{n+1}\leq 1$,
\begin{align*}
    \KL(\mu_{n+1}|p_\theta) &\leq \KL(\mu_{n}|p_\theta) + \ssp{ \log\frac{\mu_n}{p_\theta},\mu_{n+1}	-\mu_n}+\frac{1}{\gamma_{n+1}} \KL(\mu_{n+1}|\mu_n).
\end{align*}
Applying Lemma \ref{lem:three-point} to the convex function $\cG_n(\nu)=\gamma_{n+1} \ssp{\log\frac{\mu_n}{p_\theta},\nu-\mu_n}$, with $t=\mu_n$ and $\bar{z}=\mu_{n+1}$ yields
\begin{equation*}
    \ssp{\log\frac{\mu_n}{p_\theta},\mu_{n+1}	-\mu_n} + \frac{1}{\gamma_{n+1}} \KL(\mu_{n+1}|\mu_n) \le \ssp{ \log\frac{\mu_n}{p_\theta},\nu-\mu_n} + \frac{1}{\gamma_{n+1}} \KL(\nu| \mu_n) - \frac{1}{\gamma_{n+1}} \KL(\nu| \mu_{n+1}),
\end{equation*}
and thus
\begin{align}
\label{eq:conv_3point_new}
    \KL(\mu_{n+1}|p_\theta) &\leq \KL(\mu_{n}|p_\theta) + \ssp{ \log\frac{\mu_n}{p_\theta},\nu-\mu_n} + \frac{1}{\gamma_{n+1}} \KL(\nu| \mu_n) - \frac{1}{\gamma_{n+1}} \KL(\nu| \mu_{n+1}).
\end{align}
As the reverse KL with respect to any target is also $1$-relatively convex with respect to the $\KL$ \citepapp{aubin2022mirror, chopin2023connection}, we have
\begin{equation*}
    \ssp{\log\frac{\mu_n}{p_\theta},\nu-\mu_n}\le \KL(\nu|p_\theta)-\KL(\mu_{n}|p_\theta) - \KL(\nu|\mu_n) 
\end{equation*}
and \eqref{eq:conv_3point_new} becomes
\begin{equation}\label{eq:conv_strcvx_new}
\KL(\mu_{n+1}|p_\theta)\le \KL(\nu|p_\theta) + \left(\frac{1}{\gamma_{n+1}}-1\right) \KL(\nu|\mu_n) - \frac{1}{\gamma_{n+1}}\KL(\nu| \mu_{n+1}).
\end{equation}
Plugging the above into~\eqref{eq:F_decrease1} gives
\begin{align*}
    \cF(z_{n+1}) 
    &\leq \KL(\nu|p_\theta)+\left(\frac{1}{\gamma_{n+1}} - l\right)B_h(\theta| \theta_n) +\left(\frac{1}{\gamma_{n+1}}-1\right) \KL(\nu|\mu_n)\\
    &\qquad- \frac{1}{\gamma_{n+1}} \left[B_h(\theta| \theta_{n+1})+\KL(\nu| \mu_{n+1})\right].
\end{align*}
Denoting $z=(\theta, \nu)$ and recalling the definition of $B_\phi$ in Proposition~\ref{prop:derivative} we find
\begin{align*}
    \cF(z_{n+1}) 
    &\leq \cF(z)+\left(\frac{1}{\gamma_{n+1}} - \min(1, l)\right)B_\phi(z|z_{n})- \frac{1}{\gamma_{n+1}} B_\phi(z|z_{n+1}).
\end{align*}
This shows in particular, by substituting $z=z_n$ and since $B_{\phi}(z| z_{n+1})\ge 0$, that 
 \begin{align*}
     \cF(z_{n+1})\le \cF(z_n)-\frac{1}{\gamma_{n+1}} B_{\phi}(z_n| z_{n+1}),
 \end{align*}
i.e. $\cF$ is decreasing at each iteration.

Multiplying the previous equation by $(\gamma_{n+1}^{-1}-\min(1, l))^{-1}$, we get
\begin{align*}
    \left(\frac{1}{1- \gamma_{n+1} \min(1, l)}\right)[\cF(z_{n+1}) -\cF(z)]
    &\leq \frac{1}{\gamma_{n+1}}B_\phi(z|z_{n})- \frac{1}{\gamma_{n+1}}\left(\frac{1}{1- \gamma_{n+1} \min(1, l)}\right) B_\phi(z|z_{n+1}),
\end{align*}
and, proceeding as in \citetapp[Appendix A.2]{chopin2023connection} we obtain
\begin{equation}
\label{eq:md_conv}
    \cF(z_n) - \cF(z) \le \frac{C_n}{\gamma_1}  B_{\phi}(z| z_{0})
\end{equation}
where
\begin{align*}
    C_n^{-1} = \sum_{k=1}^n\frac{\gamma_k}{\gamma_1}\prod_{i=1}^k\frac{1}{1-\gamma_i \min(1, l)}.
\end{align*}

Following \citetapp[Appendix A.3]{chopin2023connection} we can then show that
\begin{align}
\label{eq:rate_bound1}
    C_n = \left(\sum_{k=1}^n\prod_{i=1}^k\frac{\gamma_k/\gamma_1}{1-\gamma_i\min(l, 1) }\right)^{-1}\le \prod_{k=1}^n(1-\gamma_k\min(l, 1))
\end{align}
by induction. Plugging $z=(\theta^\star, p_{\theta^\star}(\cdot|y))$ into~\eqref{eq:md_conv} we obtain the result.

\subsubsection{Proof of~\eqref{eq:rate_bound1}}
To see this we consider for $n \ge 1$, $\mathcal{P}(n): \quad \sum_{k=1}^{n}\frac{\gamma_k}{\gamma_1}\prod_{i=1}^k\frac{1}{1-\gamma_i\min(l, 1)} \geq \prod_{k=1}^n(1-\gamma_k)^{-1}$.
We trivially have that  $\mathcal{P}(1): \left(\frac{1}{1-\gamma_1\min(l, 1) }\right)^1 \geq (1-\gamma_1\min(l, 1))^{-1}$ is true. Then, assume $\mathcal{P}(n)$ holds.
We have 
\begin{align*}
\sum_{k=1}^{n+1} \frac{\gamma_k}{\gamma_1}\prod_{i=1}^k\frac{1}{1-\gamma_i\min(l, 1)} 
&= \sum_{k=1}^{n} \frac{\gamma_k}{\gamma_1}\prod_{i=1}^k\frac{1}{1-\gamma_i\min(l, 1)} + \frac{\gamma_{n+1}}{\gamma_1}\prod_{i=1}^{n+1}\frac{1}{1-\gamma_i\min(l, 1)} \\
&\geq \prod_{k=1}^n(1-\gamma_k\min(l, 1))^{-1} + \gamma_{n+1}\prod_{k=1}^{n+1}(1-\gamma_k\min(l, 1))^{-1}\\
&=\prod_{k=1}^n(1-\gamma_k\min(l, 1))^{-1} \left[1+ \gamma_{n+1}(1-\gamma_{n+1}\min(l, 1))^{-1}\right],
\end{align*}
since $\gamma_1 \leq 1$.
Observing that $1+ \gamma_{n+1}(1-\gamma_{n+1}\min(l, 1))^{-1}\geq (1-\gamma_{n+1}\min(l, 1))^{-1}$ we have the result

Hence \eqref{eq:rate_bound1} is true for all $n\ge 1$.

\section{On Replacing~\eqref{eq:md_update} with~\eqref{eq:md_fast}}
\label{app:ratio}
\subsection{Proof of~\eqref{eq:ratio}}
    Consider $n=1$, in this case $\mu_1\equiv \widetilde\mu_1$. For $n=2$ 
\begin{align*}
    \mu_2(x)&\propto \mu_0(x)^{(1-\gamma_1)(1-\gamma_2)}p_{\theta_0}(x, y)^{\gamma_1(1-\gamma_2)}p_{\theta_1}(x, y)^{\gamma_2}\\
    \widetilde\mu_2(x)&\propto \mu_0(x)^{(1-\gamma_1)(1-\gamma_2)}p_{\theta_1}(x, y)^{1-(1-\gamma_1)(1-\gamma_2)},
\end{align*}
and 
\begin{align*}
    \frac{\mu_{2}(x)}{\widetilde{\mu}_{2}(x)} &\propto \left(\frac{p_{\theta_0}(x, y)}{p_{\theta_1}(x, y)}\right)^{\gamma_1(1-\gamma_2)}.
\end{align*}
For $n\geq 1$ we have
\begin{align*}
    \widetilde{\mu}_{n+1}(x)&\propto \mu_0(x)^{\prod_{k=1}^{n+1}(1-\gamma_k)}p_{\theta_n}(x, y)^{1-\prod_{k=1}^{n+1}(1-\gamma_k)}\\
    &\propto\left(\frac{\widetilde{\mu}_{n}(x)}{p_{\theta_{n-1}}(x, y)^{1-\prod_{k=1}^{n}(1-\gamma_k)}}\right)^{1-\gamma_{n+1}}p_{\theta_n}(x, y)^{1-\prod_{k=1}^{n+1}(1-\gamma_k)}
\end{align*}
and
\begin{align*}
    \frac{\mu_{n+1}(x)}{\widetilde{\mu}_{n+1}(x)}\propto \left(\frac{\mu_n(x)}{\widetilde{\mu}_{n}(x)}\right)^{1-\gamma_{n+1}}\left(\frac{p_{\theta_{n-1}}(x, y)}{p_{\theta_n}(x, y)}\right)^{(1-\gamma_{n+1})(1-\prod_{k=1}^n(1-\gamma_k))}.
\end{align*}
Plugging $\frac{\mu_{n}(x)}{\widetilde{\mu}_{n}(x)}$ into above the result follows by induction using that $1-\prod_{k=1}^{n}(1-\gamma_k) = \sum_{k=1}^n\gamma_k \prod_{j=k+1}^n(1-\gamma_j)$.

\subsection{Toy Gaussian Model}
\label{app:toy}

\paragraph{Convexity and Lipschitz continuity}
To see that the toy LVM satisfies Assumption~\ref{ass:convex} consider
\begin{align*}
    p_\theta(x, y) \propto \prod_{i=1}^{d_x} \frac{1}{2\uppi}\exp\left(-\frac{(x_i-\theta)^2}{2} - \frac{(y_i-x_i)^2}{2}\right)
\end{align*}
so that
\begin{align*}
    U(\theta, x) = d_x\log(2\uppi)+\frac{1}{2}\sum_{i=1}^{d_x} (x_i-\theta)^2+(y_i-x_i)^2.
\end{align*}
Then $\nabla_\theta U(\theta, x) = d_x\theta-\sum_{i=1}^{d_x}x_i$ and
\begin{align*}
     U(\theta_2, x) -  U(\theta_1, x)- \ssp{\nabla_\theta    U(\theta_1, x), \theta_2-\theta_1} = \frac{d_x}{2}(\theta_2-\theta_1)^2,
\end{align*}
showing that $U$ is both relatively convex and relatively smooth w.r.t. the Euclidean norm with $l=L = d_x/2$.

\paragraph{Experimental set up} We set $\theta = 1$ and $d_x=50$ and generate one data point. The initial state is $\theta_0 = 0, \mu_0 = \mathcal{N}(0, \textsf{Id}_{d_x})$, we use $N=200$ particles, $T=2000$ and $\gamma_n \equiv 0.01$.

\paragraph{Additional Results}

To further confirm that for large $n$ the iterates~\eqref{eq:smc_bad} and~\eqref{eq:md_fast} are close we consider Wasserstein-1 distance between each 1 dimensional marginal of $\mu_n, \widetilde{\mu}_n$ (Figure~\ref{fig:toy2}).
As expected, as $n$ increases we have $W_1(\mu_n, \widetilde{\mu}_n)$. When comparing the posterior approximations of the 1D marginals we find that MD-LVM and SMCs-LVM provide very similar approximations.

\begin{figure}
	\centering
	\begin{tikzpicture}[every node/.append style={font=\normalsize}]
		\node(img4) {\includegraphics[width = 0.3\textwidth]{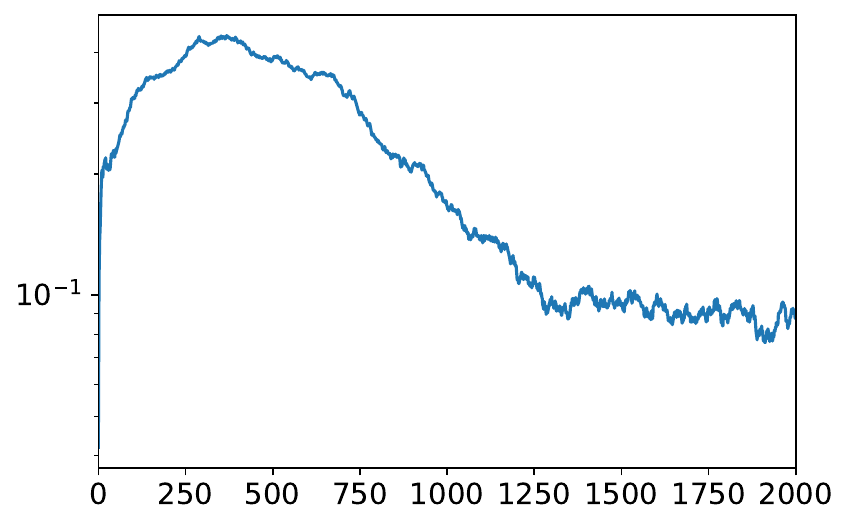}};
		\node[right=of img4, node distance = 0, xshift = -0.5cm] (img5) {\includegraphics[width = 0.3\textwidth]{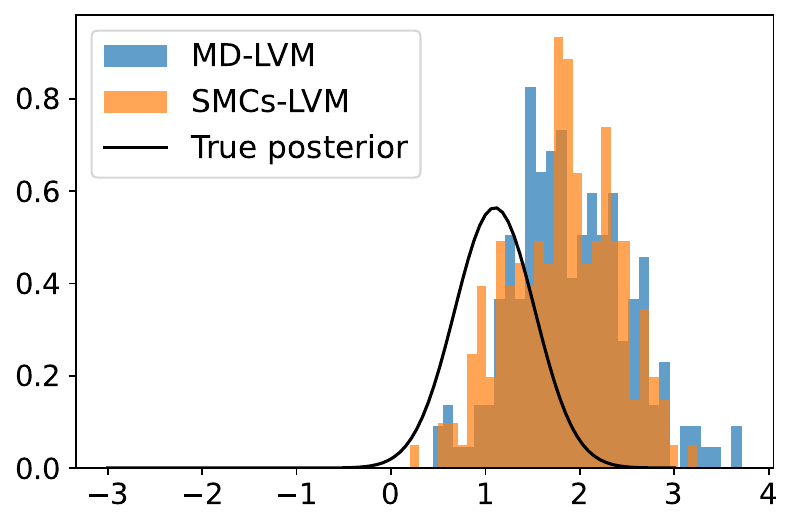}};
  \node[right=of img5, node distance = 0, xshift = -0.5cm] (img6) {\includegraphics[width = 0.3\textwidth]{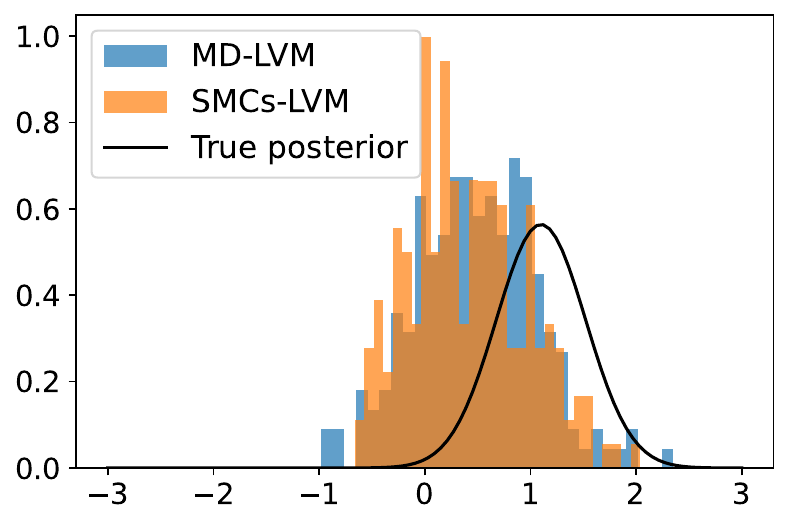}};
		\node[below=of img4, node distance = 0, yshift = 1.2cm] {$n$};
		\node[below=of img5, node distance = 0, yshift = 1.2cm] {$x$};
  		\node[below=of img6, node distance = 0, yshift = 1.2cm] {$x$};
        \node[left=of img4, node distance = 0, rotate = 90, anchor = center, yshift = -0.8cm] {$W_1(\mu_n, \widetilde{\mu}_n)$};
        \node[left=of img5, node distance = 0, rotate = 90, anchor = center, yshift = -0.8cm] {$p_{\theta_T^N}(x|y)$};
        \node[left=of img6, node distance = 0, rotate = 90, anchor = center, yshift = -0.8cm] {$p_{\theta_T^N}(x|y)$};
	\end{tikzpicture}
	\caption{Comparison of~\eqref{eq:smc_bad} and~\eqref{eq:md_fast} on a toy Gaussian model. Left: evolution of $W_1$ along iterations. Middle: final posterior approximation for marginal of component 1. Right: final posterior approximation for marginal of component 35.}
	\label{fig:toy2}
\end{figure}

\section{Proof of Section~\ref{sec:smc}}
\label{app:proof}
\subsection{Feynman-Kac model for SMC samplers}

\begin{algorithm}[th]
\begin{algorithmic}[1]
\STATE{\textit{Inputs:} sequences of distributions $(\mu_n)_{n= 0}^T$, Markov kernels $(M_n)_{n= 1}^T$, initial proposal $\mu_0$.}
\STATE{\textit{Initialise:} sample $\widetilde{X}_0^i\sim \mu_0$ and set $W_0^i=1/N$ for $i=1,\dots, N$.}
\FOR{$n=1,\dots, T$}
\IF{$n>1$}
\STATE{\textit{Resample:} draw $ \{\widetilde{X}_{n-1}^i\}_{i=1}^N$ independently from $\{X_{n-1}^i, W_{n-1}^i\}_{i=1}^N$ and set $W_n^i =1/N$ for $i=1,\dots, N$. \label{alg:reformulated:resampling}}
\ENDIF
\STATE{\textit{Propose:} draw $X_n^i\sim M_{n}(\widetilde{X}_{n-1}^i, \cdot)$ for $i=1,\dots, N$.\label{alg:reformulated:mutate}}
\STATE{\textit{Reweight:} compute and normalise the weights $W_n^i \propto w_{n}(X_{n}^i)$ in~\eqref{eq:smc_weight} for $i=1,\dots, N$.\label{alg:reformulated:reweight}}
\ENDFOR
\STATE{\textit{Output:} $\{X_n^i,W_n^i\}_{i=1}^N$}
\end{algorithmic}
\caption{SMC samplers \citep{del2006sequential}.}\label{alg:smc}
\end{algorithm}

We consider the general framework of Feynman-Kac measure flows, that is, a sequence of probability measures $( \hat{\eta}_n)_{n\geq0}$ of increasing dimension defined on Polish spaces $(E^n, \mathcal{E}^n)$, where $\mathcal{E}$ denotes the $\sigma$-field associated with $E$, which evolves as
\begin{equation}
\label{eq:smcflow}
\update(d x_{1:n}) \propto  G_n(x_{n-1}, x_{n})K_n(x_{n-1}, d x_n) \hat{\eta}_{n-1}(d x_{1:n-1}),
\end{equation}
for some Markov kernels $K_n: E\times \mathcal{E}\to [0,1]$ and non-negative functions $G_n:  E\times E\to \real$, and with $\hat{\eta}_0(d x_0) \propto G_0(x_0)K_0(d x_0)$.

Recursion~\eqref{eq:smcflow} can be decomposed into two steps. In the mutation step, a new state is proposed according to $K_n$
\begin{align}
\label{eq:predictive}
\eta_n(d x_{1:n})\propto\hat{\eta}_{n-1}(d x_{1:n-1}) K_n(x_{n-1}, d x_n);
\end{align}
in the selection step, the proposed state is weighted according to the potential function $\weight$
\begin{align}
\label{eq:update}
\hat{\eta}_n(d x_{1:n})\propto\eta_n(d x_{1:n})G_n(x_n).
\end{align}
To ease the exposition of the theoretical results of this section we introduce the Boltzmann-Gibbs operator associated with the weight function $\weight$
\begin{align}
    \label{eq:bg}
    \update(d x_{1:n}) = \bg(\predictive)(d x_{1:n}) = \frac{\predictive(d x_{1:n})\weight(x_n)}{\predictive(\weight)},
\end{align}
which weights $\predictive$ using $\weight$ and returns an appropriately normalised probability measure.

The SMC sampler in Algorithm~\ref{alg:smc} can be obtained by setting $K_0 \equiv \mu_0$, $G_0\equiv 1$, and
\begin{align*}
    K_n(x_{n-1}, d x_n) &= M_n( x_{n-1},d x_n )\\
    G_n(x_{n-1}, x_n) &= \frac{\mu_n(x_n)}{\mu_{n-1}(x_n)},
\end{align*}
as shown in \citetapp[Chapter 17]{chopin2020introduction}. The idealised versions of Algorithm~\ref{alg:smc_lvm}, in which the $(\theta_n)_{n\geq 0}$ is known and fixed can by obtained in the same say.

For convenience, we identify the three fundamental steps of Algorithm~\ref{alg:smc} as a  \emph{resampling} step (Line~\ref{alg:reformulated:resampling}), a \emph{mutation} step (Line~\ref{alg:reformulated:mutate}) and a \emph{reweighting} step (Line~\ref{alg:reformulated:reweight}). To each step, we associate a measure and its corresponding particle approximation: the mutated measure $\predictive$ in~\eqref{eq:predictive} is approximated by $\predictiveN := N^{-1} \sum_{i=1}^N \delta_{X_n^i}$ obtained after Line~\ref{alg:reformulated:mutate}, Line~\ref{alg:reformulated:reweight} provides a particle approximation of $\bg(\predictive)\equiv\update$ denoted by $\bg(\predictiveN)$,
after resampling we obtain another approximation of $\update$ in~\eqref{eq:update}, $\updateN := N^{-1} \sum_{i=1}^N \delta_{\tilde{X}_n^i}$.

\subsection{Feynman-Kac model for Algorithm~\ref{alg:smc_lvm} and SMC-LVMs}

First we observe that since the true sequence $(\theta_n)_{n\geq 0}$ is not known but approximated via~\eqref{eq:theta_update}, Algorithm~\ref{alg:smc_lvm} and SMC-LVMs in Section~\ref{sec:smc_lvm} use weight functions and Markov kernels which are random and approximate the true but unknown weight function and kernel. 
In particular, Algorithm~\ref{alg:smc_lvm} uses the approximate kernels which leave $\mu_n(\cdot; \theta_{0:n-2}^N)$ invariant and corresponding weight functions
\begin{align*}
    K_{n,N}(x_{n-1}, d x_n) &= M_n(x_{n-1}, d x_n; \theta_{0:n-2}^N)\\
    \weightN(x_n) &=w_n(x_n, \theta_{0:n-1}^N);
\end{align*}
which are approximations of the same quantities with the exact $\theta$-sequence $(\theta_n)_{n\geq 0}$.

For any distribution $\eta$ and any $\testfn\in\bounded$ we denote $\eta(\testfn):= \int \testfn(x)\eta(x)dx$, similarly for all empirical distributions $\eta^N:=N^{-1}\sum_{i=1}^N \delta_{X^i}$ we denote the corresponding average by $\eta^N(\testfn):= N^{-1}\sum_{i=1}^N \testfn(X^i)$.

\subsection{Stability of Weights}

We first show that Assumption~\ref{ass:weights} implies a stability result on the weights~\eqref{eq:md_weights}.

\begin{lemma}
    \label{lem:weight_stability}
    Under Assumption~\ref{ass:weights}, there exists a constant $\omega>0$ such that
    \begin{align}
    \label{eq:lipschitz}
    \vert w_n(x;\theta_{0:n-1}) - w_n(x;\theta_{0:n-1}')\vert \leq \omega\sum_{j=0}^{n-1}\norm{\theta_j-\theta_j'}
\end{align}
for all $(\theta_{0:n-1}, \theta_{0:n-1}')\in(\real^{d_\theta})^n$.
\end{lemma}
\begin{proof}
    Consider 
    \begin{align*}
        \gamma_n^{-1} \log w_n(x;\theta_{0:n-1}) &= -U(\theta_{n-1}, x)+{\prod_{k=1}^{n-1}(1-\gamma_k)}\log \mu_0(x)\\
        &+\gamma_{n-1} U(\theta_{n-2},x)+\dots +\prod_{k=2}^{n-1}(1-\gamma_k)U(\theta_{0}, x)
    \end{align*}
    which has gradient
    \begin{align*}
        \gamma_n^{-1} \nabla_{\theta_{0:n-1}}\log w_n(x;\theta_{0:n-1}) =\begin{pmatrix}
            \prod_{k=2}^{n-1}(1-\gamma_k)\nabla_\theta U(\theta_{0}, x)\\
            \vdots\\
            \gamma_{n-1} \nabla_\theta U(\theta_{n-2},x)\\
            -\nabla_\theta U(\theta_{n-1}, x)
        \end{pmatrix}.
    \end{align*}
    Since $\supnorm{\nabla_\theta U}<\infty $ it follows that $\gamma_n^{-1}\supnorm{ \nabla_{\theta_{0:n-1}}\log w_n} <\infty$.
    Observing that $$\nabla_{\theta_{0:n-1}} w_n(x;\theta_{0:n-1}) = w_n(x;\theta_{0:n-1})\nabla_{\theta_{0:n-1}}\log w_n(x;\theta_{0:n-1})$$ and that, under Assumption~\ref{ass:weights} $\supnorm{w_n}<\infty$ we obtain $\supnorm{\nabla_{\theta_{0:n-1}} w_n}<\infty$ from which follows the Lipschitz continuity in~\eqref{eq:lipschitz}.
    
\end{proof}

\subsection{Proof of Proposition~\ref{prop:lp}}

We proceed by induction, taking $n=0$ as the base case.
At time $n=0$, the particles $(X_0^{i})_{i=1}^N$ are sampled i.i.d. from $\mu_0$, so that $G_0(x)=G_0^N(x)\equiv 1$, hence $\Psi_{G_0}(\eta_0)\equiv  \hat{\eta}_0\equiv \mu_0$  and  $\Exp\left[\testfn(X_0^{i})\right] = \Psi_{G_0}(\eta_0)(\testfn)$ for $i=1,\ldots,N$.
We can define the sequence of functions $\Delta_0^i: \calX \mapsto \mathbb{R}$ for $i=1,\ldots, N$
\begin{equation*}
\Delta_0^i(x) :=  \testfn(x) -  \Exp\left[\testfn(X_{0}^{i})\right]
\end{equation*}
so that,
\begin{align*}
\Psi_{G_0^N}(\eta_0^N)(\testfn) - \Psi_{G_0}(\eta_0)(\testfn) = \frac{1}{N}\sum_{i=1}^N \Delta_0^i(X_0^{i}),
\end{align*}
and apply Lemma~\ref{lemma:delmoral} below to get for every $p\geq 1$
\begin{align}
\label{eq:lp0_iid}
\Exp\left[\vert \Psi_{G_0^N}(\eta_0^N)(\testfn) - \Psi_{G_0}(\eta_0)(\testfn)\vert^p\right]^{1/p} &\leq b(p) ^{1/p} \frac{1}{\sqrt{N}} \left(\sum_{i=1}^N \left(\sup( \Delta_{0}^i) - \inf( \Delta_{0}^i)\right)^2\right)^{1/2} \\
&\leq b(p) ^{1/p} \frac{1}{\sqrt{N}} \left(\sum_{i=1}^N 4\left(\sup\vert \Delta_{0}^i\vert\right)^2\right)^{1/2} \notag\\
& \leq b(p) ^{1/p} \frac{1}{\sqrt{N}} \left(\sum_{i=1}^N 16\supnorm{\testfn}^2\right)^{1/2} \notag \\
& \leq 4b(p) ^{1/p} \frac{1}{\sqrt{N}}\supnorm{\testfn},\notag
\end{align}
with $C_{p, 0} = 4b(p)^{1/p}.$
Since $h=\norm{\cdot}^2/2$ and assuming the same initial value of the $\theta$-iterates is fixed, we have
\begin{align*}
    \Exp[\norm{\theta_1^N-\theta_1}^p]^{1/p}&\leq \gamma_1\Exp\left[\vert \Psi_{G_0^N}(\eta_0^N)(\nabla_\theta U(\theta_0^N, \cdot)) - \Psi_{G_0}(\eta_0)(\nabla_\theta U(\theta_0^N, \cdot))\vert^p\right]^{1/p}\\
    &\leq \gamma_1\Exp\left[\vert \Psi_{G_0^N}(\eta_0^N)(\nabla_\theta U(\theta_0^N, \cdot) -\nabla_\theta U(\theta_0, \cdot))\vert^p\right]^{1/p}\\
    &+\gamma_1\Exp\left[\vert \Psi_{G_0^N}(\eta_0^N)(\nabla_\theta U(\theta_0, \cdot))-\Psi_{G_0}(\eta_0)(\nabla_\theta U(\theta_0, \cdot))\vert^p\right]^{1/p}\\
    &\leq 4b(p) ^{1/p} \frac{\gamma_1}{\sqrt{N}}\supnorm{\nabla_\theta U},
\end{align*}
using~\eqref{eq:lp0_iid} and recalling that $\theta_0=\theta_0^N$.
Hence, $D_{p, 1} = 4b(p)^{1/p}\supnorm{\nabla_\theta U}$.

Then, assume that the result holds for all times up to time $n-1$ for some $n$: we will show it also holds at time $n$.
Using Lemma~\ref{lp:lemma3}, which controls the error of the reweighting step, we have
\begin{align*}
\Exp\left[\vert\bgN(\predictiveN)(\testfn) - \bg(\predictive)(\testfn)\vert^p\right]^{1/p} 
&\leq \frac{2\omega\supnorm{\testfn}}{ \predictive( \weight)}\sum_{j=0}^{n-1}\Exp\left[\Vert \theta_{j}^N-\theta_{j}\Vert^p\right]^{1/p}\\
&+\frac{\supnorm{\testfn}}{\predictive(\weight)}\Exp\left[\vert\predictiveN(G_n) - \predictive(G_n)\vert^p\right]^{1/p}\\
&+\frac{1}{\predictive(\weight)}\Exp\left[\vert\predictiveN(G_n\testfn) - \predictive(G_n\testfn)\vert^p\right]^{1/p}.
\end{align*}
Applying Lemma~\ref{lp:lemma1}, controlling the error introduced by the mutation step, to the last two term above we find
\begin{align*}
\Exp\left[\vert\bgN(\predictiveN)(\testfn) - \bg(\predictive)(\testfn)\vert^p\right]^{1/p} 
&\leq \frac{2\omega\supnorm{\testfn}}{ \predictive( \weight)}\sum_{j=0}^{n-1}\Exp\left[\Vert \theta_{j}^N-\theta_{j}\Vert^p\right]^{1/p}\\
&+\frac{2\supnorm{\testfn}\supnorm{G_n}}{ \predictive( \weight)}\left(\sum_{j=0}^{n-2}\Exp\left[\Vert \theta_{j}^N-\theta_{j}\Vert^p\right]^{1/p} +4\frac{b(p)^{1/p}}{N^{1/2}}\right)\\
&+\frac{\supnorm{\testfn}}{\predictive(\weight)}\Exp\left[\vert\hat{\eta}_{n-1}^NK_n(G_n) - \hat{\eta}_{n-1}K_n(G_n)\vert^p\right]^{1/p}\\
&+\frac{1}{\predictive(\weight)}\Exp\left[\vert\hat{\eta}_{n-1}^NK_n(G_n\testfn) - \hat{\eta}_{n-1}K_n(G_n\testfn)\vert^p\right]^{1/p}.
\end{align*}
Finally, Lemma~\ref{lp:lemma4} controls the error introduced by the resampling step and gives
\begin{align*}
&\Exp\left[\vert\bgN(\predictiveN)(\testfn) - \bg(\predictive)(\testfn)\vert^p\right]^{1/p}\\ 
&\qquad\qquad\leq \frac{2\omega\supnorm{\testfn}}{ \predictive( \weight)}\sum_{j=0}^{n-1}\Exp\left[\Vert \theta_{j}^N-\theta_{j}\Vert^p\right]^{1/p}\\
&\qquad\qquad+\frac{2\supnorm{\testfn}\supnorm{G_n}}{ \predictive( \weight)}\left(\sum_{j=0}^{n-2}\Exp\left[\Vert \theta_{j}^N-\theta_{j}\Vert^p\right]^{1/p} +8\frac{b(p)^{1/p}}{N^{1/2}}\right)\\
&\qquad\qquad+\frac{\supnorm{\testfn}}{\predictive(\weight)}\Exp\left[\vert\Psi_{G_{n-1}^N}(\eta_{n-1}^N)(K_n(G_n)) - \Psi_{G_{n-1}}(\eta_{n-1})(K_n(G_n))\vert^p\right]^{1/p}\\
&\qquad\qquad+\frac{1}{\predictive(\weight)}\Exp\left[\vert\Psi_{G_{n-1}^N}(\eta_{n-1}^N)(K_n(G_n\testfn)) - \Psi_{G_{n-1}}(\eta_{n-1})(K_n(G_n\testfn))\vert^p\right]^{1/p}.
\end{align*}
Recalling that $\theta_0^N=\theta_0$ and using the results for all times from 1 to $n-1$ and the fact that $\supnorm{K_n\testfn}\leq\supnorm{\testfn} $ for all $\testfn\in\bounded(\calX)$, we find
\begin{align*}
\Exp\left[\vert\bgN(\predictiveN)(\testfn) - \bg(\predictive)(\testfn)\vert^p\right]^{1/p} 
&\leq \frac{2\omega\supnorm{\testfn}}{ \predictive( \weight)}\frac{\sum_{j=1}^{n-1} D_{p, j}\gamma_j}{N^{1/2}}+\frac{2\supnorm{\testfn}\supnorm{G_n}}{ \predictive( \weight)}\frac{\sum_{j=1}^{n-2} D_{p, j}\gamma_j}{N^{1/2}} \\
&+\frac{16\supnorm{\testfn}\supnorm{G_n}}{ \predictive( \weight)}\frac{b(p)^{1/p}}{N^{1/2}}+\frac{2\supnorm{\testfn}\supnorm{G_n}}{\predictive(\weight)}\frac{C_{p,n-1}}{N^{1/2}}.
\end{align*}
It follows that
\begin{align*}
    C_{p, n} = \frac{2\omega\sum_{j=1}^{n-1} D_{p, j}\gamma_j}{ \predictive( \weight)}+\frac{2\supnorm{G_n}\sum_{j=1}^{n-2} D_{p, j}\gamma_j}{ \predictive( \weight)}+\frac{16b(p)^{1/p}\supnorm{G_n}}{ \predictive( \weight)}+\frac{2C_{p,n-1}\supnorm{G_n}}{\predictive(\weight)}.
\end{align*}
Proceeding similarly for $\theta$, we find, using Lemma~\ref{lp:lemma_theta},
\begin{align*}
    \Exp\left[\Vert \theta_{n}^N-\theta_{n}\Vert^p\right]^{1/p} &\leq (1+\gamma_nL)\Exp\left[\Vert\theta_{n-1}^N-\theta_{n-1}\Vert^p \right]^{1/p} \\
    &+ \gamma_n\Exp\left[\Vert \Psi_{G_{n-1}^N}(\eta_{n-1}^N)(\nabla_\theta U(\theta_{n-1}, \cdot)) - \Psi_{G_{n-1}}(\eta_{n-1})(\nabla_\theta U(\theta_{n-1}, \cdot))\Vert^p \right]^{1/p}\\
    &\leq (1+\gamma_nL) D_{p,n-1}\frac{\gamma_{n-1}}{N^{1/2}}+\supnorm{\nabla_\theta U}C_{p,n-1}\frac{\gamma_n}{N^{1/2}}\\
    &\leq D_{p, n}\frac{\gamma_n }{N^{1/2}},
\end{align*}
where we used the fact that $\gamma_{n-1}\leq \gamma_n\leq 1$ and
\begin{align*}
    D_{p,n}= (1+L) D_{p,n-1}+\supnorm{\nabla_\theta U}C_{p,n-1}.
\end{align*}
The result follows for all $n\in \mathbb{N}$ by induction.

\subsubsection{Auxiliary results for the proof of Proposition~\ref{prop:lp}}

As a preliminary we reproduce part of \citeapp[Lemma 7.3.3]{smc:theory:Del04}, a Marcinkiewicz-Zygmund-type inequality of which we will make extensive use.

\begin{lemma}[Del Moral, 2004] \label{lemma:delmoral}
Given a sequence of probability measures $(\mu_i)_{i \geq 1}$ on a given measurable space $(E,\mathcal{E})$ and a collection of independent random variables, one distributed according to each of those measures, $(X_i)_{i \geq 1}$, where $\forall i, X_i \sim \mu_i$, together with any sequence of measurable functions $(f_i)_{i \geq 1}$ such that $\mu_i(f_i) = 0$ for all $i \geq 1$, we define for any $N \in \mathbb{N}$,
$$m_N(X)(f) = \frac{1}{N} \sum_{i=1}^N f_i( X_i ) \ \textrm{ and } \ \sigma_N^2(f) = \frac{1}{N} \sum_{i=1}^N \left(\sup(f_i) - \inf(f_i) \right)^2.$$
If the $f_i$ have finite oscillations (i.e., $\sup(f_i)-\inf(f_i)<\infty \  \forall i \geq 1$) then we have:
$$\sqrt{N} \Exp\left[ \left\vert m_N(X)(f) \right\vert^p \right]^{1/p} \leq b(p)^{1/p} \sigma_N(h),$$
with, for any pair of integers $q,p$ such that $q \geq p \geq 1$, denoting $(q)_p=q!/(q-p)!$:
\begin{align}
    \label{eq:b(p)}
    b(2q) = (2q)_q 2^{-q} \ \textrm{ and } \  b(2q - 1) = \frac{(2q-1)_q}{\sqrt{q-\frac12}} 2^{-(q-\frac12)}. 
\end{align}
\end{lemma}

We also define the following $\sigma$-fields of which we will make frequent use:  $\mathcal{G}_0^N:= \sigma\left(\widetilde{X}_0^{i}: i\in\lbrace 1,\ldots,N\rbrace\right)$. We recursively define the $\sigma$-field generated by the weighted samples up to an including mutation at time $n$, $\mathcal{F}_{n}^N:= \sigma\left(X_n^i: i\in\lbrace 1,\ldots,N\rbrace\right)\vee \sfmutation$ and the $\sigma$-field generated by the particle system up to (and including) time $n$ before the mutation step at time $n+1$, $\mathcal{G}_n^N:= \sigma\left(\widetilde{X}_n^{i}: i\in\lbrace 1,\ldots,N\rbrace\right)\vee \mathcal{F}_{n}^N$.

We start by controlling the error in the $\theta$-iterates.
\begin{lemma}[$\theta$-update]
\label{lp:lemma_theta}
Under the conditions of Proposition~\ref{prop:lp}, we have that
\begin{align*}
    \Exp\left[\Vert \theta_{n}^N-\theta_{n}\Vert^p\right]^{1/p} &\leq (1+\gamma_nL)\Exp\left[\Vert\theta_{n-1}^N-\theta_{n-1}\Vert^p \right]^{1/p} \\
    &+ \gamma_n\Exp\left[\Vert \Psi_{G_{n-1}^N}(\eta_{n-1}^N)(\nabla_\theta U(\theta_{n-1}, \cdot)) - \Psi_{G_{n-1}}(\eta_{n-1})(\nabla_\theta U(\theta_{n-1}, \cdot))\Vert^p \right]^{1/p},
\end{align*}
for all $p\geq 1$.
\end{lemma}
\begin{proof}
Consider the $\theta$-update
\begin{align*}
    \theta_{n}^N = \theta_{n-1}^N - \gamma_{n}\sum_{i=1}^N W_{n-1}^i \nabla_\theta U(\theta_{n-1}^N, X_{n-1}^i) = \theta_{n-1}^N -\gamma_{n} \Psi_{G_{n-1}^N}(\eta_{n-1}^N)(\nabla_\theta U(\theta_{n-1}^N, \cdot)).
\end{align*}
Then, 
\begin{align*}
\Exp\left[\Vert \theta_{n}^N-\theta_{n}\Vert^p\right]^{1/p} &\leq \Exp\left[\Vert\theta_{n-1}^N-\theta_{n-1}\Vert^p \right]^{1/p}\\
&+\gamma_n\Exp\left[\Vert \Psi_{G_{n-1}^N}(\eta_{n-1}^N)(\nabla_\theta U(\theta_{n-1}^N, \cdot)) - \Psi_{G_{n-1}}(\eta_{n-1})(\nabla_\theta U(\theta_{n-1}, \cdot))\Vert^p \right]^{1/p}.
\end{align*}
Using the relative smoothness of $U$ in Assumption~\ref{ass:convex} we have
\begin{align*}
&\Exp\left[\Vert \Psi_{G_{n-1}^N}(\eta_{n-1}^N)(\nabla_\theta U(\theta_{n-1}^N, \cdot)) - \Psi_{G_{n-1}}(\eta_{n-1})(\nabla_\theta U(\theta_{n-1}, \cdot))\Vert^p \right]^{1/p} \\
&\qquad\qquad\leq \Exp\left[\Vert \Psi_{G_{n-1}^N}(\eta_{n-1}^N)(\nabla_\theta U(\theta_{n-1}^N, \cdot) -\nabla_\theta U(\theta_{n-1}, \cdot))\Vert^p \right]^{1/p}\\
&\qquad\qquad+\Exp\left[\Vert \Psi_{G_{n-1}^N}(\eta_{n-1}^N)(\nabla_\theta U(\theta_{n-1}, \cdot)) - \Psi_{G_{n-1}}(\eta_{n-1})(\nabla_\theta U(\theta_{n-1}, \cdot))\Vert^p \right]^{1/p}\\
&\qquad\qquad\leq L\Exp\left[\Vert\theta_{n-1}^N-\theta_{n-1}\Vert^p \right]^{1/p} \\
&\qquad\qquad+\Exp\left[\Vert \Psi_{G_{n-1}^N}(\eta_{n-1}^N)(\nabla_\theta U(\theta_{n-1}, \cdot)) - \Psi_{G_{n-1}}(\eta_{n-1})(\nabla_\theta U(\theta_{n-1}, \cdot))\Vert^p \right]^{1/p}.
\end{align*}
Combining the two results above we obtain
\begin{align*}
    \Exp\left[\Vert \theta_{n}^N-\theta_{n}\Vert^p\right]^{1/p} &\leq (1+\gamma_nL)\Exp\left[\Vert\theta_{n-1}^N-\theta_{n-1}\Vert^p \right]^{1/p} \\
    &+ \gamma_n\Exp\left[\Vert \Psi_{G_{n-1}^N}(\eta_{n-1}^N)(\nabla_\theta U(\theta_{n-1}, \cdot)) - \Psi_{G_{n-1}}(\eta_{n-1})(\nabla_\theta U(\theta_{n-1}, \cdot))\Vert^p \right]^{1/p}.
\end{align*}
\end{proof}

We now turn to the approximation error of the $\mu$-iterates.
Lemma~\ref{lp:lemma4} is a well-known result for standard SMC methods and we report it for completeness while Lemma~\ref{lp:lemma1} and~\ref{lp:lemma3} control the additional error introduced by the use of the approximate Markov kernels and weights. 

\begin{lemma}[Multinomial resampling]
\label{lp:lemma4}
Under the conditions of Proposition~\ref{prop:lp}, for any $\testfn\in \bounded$ and $p\geq 1$ we have
\begin{align*}
\Exp\left[\vert\hat{\eta}_{n-1}^N(\testfn)  - \hat{\eta}_{n-1}(\testfn)\vert^p\right]^{1/p} & \leq 4b(p)^{1/p}\frac{\supnorm{\testfn}}{N^{1/2} }++\Exp\left[\vert\Psi_{G_{n-1}^N}(\eta_{n-1}^N)(\testfn) - \hat{\eta}_{n-1}(\testfn)\vert^p\right]^{1/p}.
\end{align*}
\end{lemma}
\begin{proof}
The proof follows that of \citetapp[Lemma 5]{smc:theory:CD02}.
Divide into two terms and apply Minkowski's inequality
\begin{align*}
\Exp\left[\vert\hat{\eta}_{n-1}^N(\testfn) - \hat{\eta}_{n-1}(\testfn)\vert^p\right]^{1/p} & \leq \Exp\left[\vert\hat{\eta}_{n-1}^N(\testfn) - \Psi_{G_{n-1}^N}(\eta_{n-1}^N)(\testfn)\vert^p\right]^{1/p}\\
&+\Exp\left[\vert\Psi_{G_{n-1}^N}(\eta_{n-1}^N)(\testfn) - \hat{\eta}_{n-1}(\testfn)\vert^p\right]^{1/p}.
\end{align*}
Denote by $\sfresampling$
the $\sigma$-field generated by the weighted samples up to (and including) time $n$, $\sfresampling:= \sigma\left(X_{n-1}^i: i\in\lbrace 1,\ldots,N\rbrace\right)\vee\mathcal{G}_{n-2}^N$ and consider the sequence of functions $\Delta_n^i : E \mapsto \mathbb{R}$, $
\Delta_n^i(x) :=  \testfn(x) -  \Exp\left[\testfn(\widetilde{X}_{n-1}^i) \mid \sfresampling\right]$,  for $i=1,\ldots, N$.
Conditionally on $\sfresampling$, $\Delta_n^i(\widetilde{X}_{n-1}^i)$ $i=1,\ldots, N$ are independent and have expectation equal to 0, moreover
\begin{align*}
\hat{\eta}_{n-1}^N(\testfn) - \Psi_{G_{n-1}^N}(\eta_{n-1}^N)(\testfn) &= \frac{1}{N} \sum_{i=1}^N \left( \testfn(\widetilde{X}_{n-1}^i) - \Exp\left[\testfn(\widetilde{X}_{n-1}^i) \mid \sfresampling\right]\right)= \frac{1}{N} \sum_{i=1}^N \Delta_n^i(\widetilde{X}_{n-1}^i).
\end{align*}
Using Lemma~\ref{lemma:delmoral}, we find
\begin{align*}
\sqrt{N}\Exp\left[\vert\hat{\eta}_{n-1}^N(\testfn) - \Psi_{G_{n-1}^N}(\eta_{n-1}^N)(\testfn)\vert^p \mid \sfresampling\right]^{1/p}
&\leq b(p) ^{1/p} \frac{1}{\sqrt{N}} \left(\sum_{i=1}^N \left(\sup( \Delta_{n}^i) - \inf( \Delta_{n}^i)\right)^2\right)^{1/2}  \\
&\leq b(p) ^{1/p} \frac{1}{\sqrt{N}} \left(\sum_{i=1}^N 4\left(\sup\vert \Delta_{n}^i\vert\right)^2\right)^{1/2} \notag \\
&  \leq b(p) ^{1/p} \frac{1}{\sqrt{N}} \left(\sum_{i=1}^N 16\supnorm{\testfn}^2\right)^{1/2} \notag \\
& \leq 4b(p) ^{1/p} \supnorm{\testfn},\notag
\end{align*}
where $b(p)$ is as in~\eqref{eq:b(p)}. 
Since $\hat{\eta}_{n-1}(\testfn)\equiv \Psi_{G_{n-1}}(\eta_n)(\testfn)$, we find
\begin{align*}
\Exp\left[\vert\hat{\eta}_{n-1}^N(\testfn)  - \hat{\eta}_{n-1}(\testfn)\vert^p\right]^{1/p} & \leq 4b(p)^{1/p}\frac{\supnorm{\testfn}}{N^{1/2} }++\Exp\left[\vert\Psi_{G_{n-1}^N}(\eta_{n-1}^N)(\testfn) - \hat{\eta}_{n-1}(\testfn)\vert^p\right]^{1/p}.
\end{align*}
\end{proof}

Here we show that the mutation step preserves the error bounds; in the case of a fixed $\theta$-sequence, $K_{n,N}$ coincides with $K_n$ and essentially this result can be found in \citetapp[Lemma 3]{smc:theory:CD02}.

\begin{lemma}[Mutation]
\label{lp:lemma1}
Under the conditions of Proposition~\ref{prop:lp}, for any $\testfn\in \bounded$, $p\geq 1$ we have
\begin{align*}
\Exp\left[\vert\predictiveN(\testfn) - \predictive(\testfn)\vert^p\right]^{1/p} &\leq 4b(p)^{1/p}\frac{\supnorm{\testfn}}{ N^{1/2}}+\supnorm{\testfn}\sum_{j=0}^{n-2}\Exp\left[ \norm{\theta_j^N-\theta_j}^p\right]^{1/p}\\
&+\Exp\left[\vert\hat{\eta}_{n-1}^N K_{n}(\testfn)- \hat{\eta}_{n-1} K_{n}(\testfn) \vert^p\right]^{1/p}.
\end{align*}
\end{lemma}
\begin{proof}
Divide into three terms and apply Minkowski's inequality
\begin{align}
\label{eq:mutation_decomposition}
\Exp\left[\vert\predictiveN(\testfn) - \predictive(\testfn)\vert^p\right]^{1/p} & = \Exp\left[\vert\predictiveN(\testfn) - \hat{\eta}_{n-1} K_{n}(\testfn)\vert^p\right]^{1/p}\\
& \leq \Exp\left[\vert\predictiveN(\testfn) - \hat{\eta}_{n-1}^N K_{n,N}(\testfn)\vert^p\right]^{1/p}\notag\\
&+ \Exp\left[\vert\hat{\eta}_{n-1}^N K_{n,N}(\testfn)- \hat{\eta}^N_{n-1} K_{n}(\testfn) \vert^p\right]^{1/p}\notag\\
&+ \Exp\left[\vert\hat{\eta}_{n-1}^N K_{n}(\testfn)- \hat{\eta}_{n-1} K_{n}(\testfn) \vert^p\right]^{1/p}.\notag
\end{align}
Let $\sfmutation$ denote the $\sigma$-field generated by the particle system up to (and including) time $n-1$ before the mutation step at time $n$, $\sfmutation = \sigma\left(\widetilde{X}_p^i: i\in\lbrace 1,\ldots,N\rbrace\right)\vee \mathcal{F}_{n-1}^N$ and consider the sequence of functions $\Delta_{n}^i : E \mapsto \real$ for $i=1,\ldots, N$, $\Delta_{n}^i(x) := \testfn(x) -  \Exp\left[\testfn(X_n^i)\mid \sfmutation \right] = \testfn(x) - K_{n,N}\testfn(\widetilde{X}_{n-1}^{i})$.
Conditionally on $\sfmutation$, $\Delta_{n}^i(X^{i}_{n}),\ i=1,\ldots, N$ are independent and have expectation equal to 0, moreover
\begin{align*}
\predictiveN(\testfn) - \hat{\eta}_{n-1}^N K_{n,N}(\testfn) &= \frac{1}{N}\sum_{i=1}^N \left[ \testfn(X_{n}^{i}) -  K_{n,N}\testfn(\widetilde{X}_{n-1}^{i})\right] =\frac{1}{N} \sum_{i=1}^N \Delta_{n}^i(X_{n}^{i}).
\end{align*}
Conditioning on $\sfmutation$ and applying Lemma~\ref{lemma:delmoral} we have, for all $p\geq 1$,
\begin{align}
\label{eq:lp_mutation}
\sqrt{N}\Exp\left[\vert\predictiveN(\testfn) - \hat{\eta}_{n-1}^N K_{n,N}(\testfn)\vert^p\mid \sfmutation\right]^{1/p} &
 \leq 4b(p) ^{1/p} \supnorm{\testfn},
\end{align}
with $b(p)$ as in~\eqref{eq:b(p)}.

For the second term of the decomposition~\eqref{eq:mutation_decomposition} we use the stability of the kernels $K_{n,N}, K_n$ in Assumption~\ref{ass:kernel}
\begin{align*}
    &\vert\hat{\eta}_{n-1}^N K_{n,N}(\testfn)- \hat{\eta}^N_{n-1} K_{n}(\testfn) \vert \\
    &\qquad\qquad=\vert\frac{1}{N}\sum_{i=1}^N[K_{n,N}-K_{n}](\testfn)(\widetilde{X}_{n-1}^i) \vert\\
    &\qquad\qquad\leq \supnorm{\testfn}\sum_{j=0}^{n-2}\norm{\theta_j^N-\theta_j} .
\end{align*}

This result, Minkowski's inequality, \eqref{eq:lp_mutation} give
\begin{align*}
\Exp\left[\vert\predictiveN(\testfn) - \predictive(\testfn)\vert^p\right]^{1/p} &\leq 4b(p)^{1/p}\frac{\supnorm{\testfn}}{ N^{1/2}}+\supnorm{\testfn}\sum_{j=0}^{n-2}\Exp\left[ \norm{\theta_j^N-\theta_j}^p\right]^{1/p}\\
&+\Exp\left[\vert\hat{\eta}_{n-1}^N K_{n}(\testfn)- \hat{\eta}_{n-1} K_{n}(\testfn) \vert^p\right]^{1/p}.
\end{align*}
\end{proof}

Using the stability of the weight function in Assumption~\ref{ass:weights} and following \citetapp[Lemma 4]{smc:theory:CD02} we obtain an error bound for the approximate reweighting.
\begin{lemma}[Reweighting]
\label{lp:lemma3}
Under the conditions of Proposition~\ref{prop:lp}, for any $\testfn\in \bounded$ and $p\geq 1$ we have
\begin{align*}
\Exp\left[\vert\bgN(\predictiveN)(\testfn) - \bg(\predictive)(\testfn)\vert^p\right]^{1/p} 
&\leq \frac{2\supnorm{\testfn}}{ \predictive( \weight)}\sum_{j=0}^{n-1}\Exp\left[\Vert \theta_{j}^N-\theta_{j}\Vert^p\right]^{1/p}\\
&+\frac{\supnorm{\testfn}}{\predictive(\weight)}\Exp\left[\vert\predictiveN(G_n) - \predictive(G_n)\vert^p\right]^{1/p}\\
&+\frac{1}{\predictive(\weight)}\Exp\left[\vert\predictiveN(G_n\testfn) - \predictive(G_n\testfn)\vert^p\right]^{1/p}.
\end{align*}
\end{lemma}
\begin{proof}
Apply the definition of $\bg$ and $\bgN$ and consider the following decomposition
\begin{align*}
\vert\bgN(\predictiveN)(\testfn) - \bg(\predictiveN)(\testfn)\vert & = \left\lvert \frac{\predictiveN( \weightN\testfn)}{\predictiveN( \weightN)} - \frac{\predictive( \weight\testfn)}{\predictive( \weight)} \right\rvert\\
&\leq \left\lvert \frac{\predictiveN( \weightN\testfn)}{\predictiveN( \weightN)} - \frac{\predictiveN( \weightN\testfn)}{\predictive( \weight)} \right\rvert + \left\lvert \frac{\predictiveN( \weightN\testfn)}{\predictive( \weight)} - \frac{\predictive( \weight\testfn)}{\predictive( \weight)} \right\rvert.
\end{align*}
Then, for the first term
\begin{align*}
\left\lvert \frac{\predictiveN( \weightN\testfn)}{\predictiveN( \weightN)} - \frac{\predictiveN( \weightN \testfn)}{\predictive( \weight)} \right\rvert & = \left\lvert \frac{\predictiveN( \weightN\testfn)}{\predictiveN( \weightN)}\right\rvert \left\lvert \frac{\predictive(\weight) - \predictiveN(\weightN)}{\predictive( \weight)} \right\rvert\\
&\leq \frac{\supnorm{\testfn}}{\vert \predictive( \weight) \vert} \vert \predictive(\weight) - \predictiveN(\weightN)\vert\\
&\leq \frac{\supnorm{\testfn}}{\vert \predictive( \weight) \vert} \vert \predictive(\weight) - \predictiveN(\weight)\vert+\frac{\supnorm{\testfn}}{\vert \predictive( \weight) \vert} \vert \predictiveN(\weight) - \predictiveN(\weightN)\vert.
\end{align*}
For the second term
\begin{align*}
\left\lvert \frac{\predictiveN( \weightN\testfn)}{\predictive( \weight)} - \frac{\predictive( \weight\testfn)}{\predictive( \weight)} \right\rvert & = \frac{1}{\vert \predictive( \weight) \vert}  \vert \predictiveN( \weightN\testfn) - \predictive( \weight\testfn)\vert\\
&\leq \frac{1}{\vert \predictive( \weight) \vert}  \vert \predictiveN( \weightN\testfn) - \predictiveN( \weight\testfn)\vert+\frac{1}{\vert \predictive( \weight) \vert}  \vert \predictiveN( \weight\testfn) - \predictive( \weight\testfn)\vert.
\end{align*}
Using Assumption~\ref{ass:weights} and Lemma~\ref{lem:weight_stability} have
\begin{align*}
    \vert \predictiveN( \weightN\testfn) - \predictiveN( \weight\testfn)\vert &\leq \frac{1}{N}\sum_{i=1}^N\vert\testfn(X_n^i)[\weightN(X_n^i)-\weight(X_n^i)]\vert\\
    &\leq \supnorm{\varphi}\omega\sum_{j=0}^{n-1} \norm{\theta_j^N-\theta_j}.
\end{align*}
Combining the above with Minkowski's inequality, we have
\begin{align*}
\Exp\left[\vert\bgN(\predictiveN)(\testfn) - \bg(\predictive)(\testfn)\vert^p\right]^{1/p} 
&\leq \frac{2\omega\supnorm{\testfn}}{ \predictive( \weight)}\sum_{j=0}^{n-1}\Exp\left[\Vert \theta_{j}^N-\theta_{j}\Vert^p\right]^{1/p}\\
&+\frac{\supnorm{\testfn}}{\predictive(\weight)}\Exp\left[\vert\predictiveN(G_n) - \predictive(G_n)\vert^p\right]^{1/p}\\
&+\frac{1}{\predictive(\weight)}\Exp\left[\vert\predictiveN(G_n\testfn) - \predictive(G_n\testfn)\vert^p\right]^{1/p}.
\end{align*}

\end{proof}

\subsection{Proof of Corollary~\ref{cor:parameter}}

We start by observing that, under Assumption~\ref{ass:convex} with $l>0$, $U(\cdot, x)$ is a strongly convex function of $\theta$ uniformly in $x$ and thus by a form of the Prékopa-Leindler inequality for strong convexity \citeapp[Theorem 3.8]{saumard2014log} the marginal likelihood $p_\theta(y) = \int_{\calX} \exp\left(-U(\theta, x)\right) dx$ is strongly log-concave and thus admits a unique maximiser $\theta^\star$ and the corresponding posterior $p_\theta^\star(\cdot|y)$ is also unique.

Then we can decompose
\begin{align*}
    \Exp[\norm{\theta_n^N-\theta^\star}^2]^{1/2} \leq \Exp[\norm{\theta_n-\theta^\star}^2]^{1/2}+\Exp[\norm{\theta_n-\theta_n^N}^2]^{1/2},
\end{align*}
where we can use Proposition~\ref{prop:lp} to bound the second term.

For the first term we exploit the inequalities established in \citeapp{caprio2024error}. Assumption~\ref{ass:convex} with $L, l>0$ implies Assumption 3 and 4 therein. In addition, Assumption~\ref{ass:convex} and \citeapp[Remark 1]{akyildiz2023interacting} implies that $\theta \mapsto p_\theta(y)$ is differentiable and so is $\theta \mapsto p_\theta(\cdot |y)$.

Then, since we assumed that $p_\theta(\cdot, y)>0$ for all $(\theta, x)\in \real^{d_\theta}\times \calX$ and that $\theta \mapsto p_\theta(\cdot |y)$ is twice differentiable we can apply \citetapp[Theorem 4, Theorem 2]{caprio2024error} which give the following bound 
\begin{align*}
  \mathcal{F}(\theta_n, \mu_n) - \log p_\theta^\star(y) &\geq \frac{l}{2}\left(\norm{\theta_n-\theta^\star}^2+ W_2(\mu_n, p_\theta^\star(\cdot|y))^2\right) \\
  &\geq \frac{l}{2}\norm{\theta_n-\theta^\star}^2,
\end{align*}
where $W_2$ denotes the Wasserstein distance between $\mu_n$ and the posterior. Using Corollary~\ref{cor:convergence} we finally have
\begin{align*}
    \norm{\theta_n-\theta^\star}^2 \leq \frac{2}{l}\frac{\KL(p_{\theta^\star}(\cdot|y)|\mu_0)+\norm{\theta^\star-\theta_0}^2}{\gamma_1} \prod_{k=1}^n(1-\gamma_k\min(l, 1))
\end{align*}
Combining this with Proposition~\ref{prop:lp} with $p=2$ gives the result.
\section{Additional details for the numerical experiments}



\subsection{Gaussian Mixture}
\label{app:mix}
\paragraph{Model specification}
For this model we have
\begin{align*}
    p_\theta(x)=p(x) = \begin{cases}
        1\qquad\qquad\qquad \text{w.p. }\alpha\\
        -1\qquad\qquad \text{w.p. }1-\alpha
    \end{cases},\qquad\qquad p_\theta(y|x) = \mathcal{N}(y; x\cdot\theta, 1).
\end{align*}
It follows that 
\begin{align*}
    U(\theta, x) = -\log \alpha + 0.5(y- x\theta)^2+0.5\log 2\uppi.
\end{align*}
\paragraph{Convexity and Lipschitz continuity}

It is easy to check that for both $x=1$ and $x=-1$ we have
\begin{align*}
    \nabla_\theta U(\theta, x) = -(y-x\theta)x
\end{align*}
from which we obtain
\begin{align*}
     U(\theta_2, x) -  U(\theta_1, x)- \ssp{\nabla_\theta    U(\theta_1, x), \theta_2-\theta_1} = x^2(\theta_2-\theta_1)^2,
\end{align*}
showing that $U$ is only locally smooth w.r.t. the Euclidean distance as $L=x^2$ but convex since $l=0$.

\paragraph{Experimental set up and further numerical results}
We simulate $1000$ data points from the model with $\theta=1$, consider $\alpha\in [0.5, 1]$ and select $\theta_0=-2$. For SMCs-LVM we set $\mu_0$ to be uniform over $\{-1,1\}$.
The initial distribution for the latent variables is given by an equal probability of allocation to each of the two components, we select $N=1000$, $\gamma_n\equiv 0.05$ and iterate for $T=300$ steps.

\subsection{Multimodal example}
\label{app:multimodal}
\paragraph{Model specification}
For this model we have
\begin{align*}
    U(\theta, x) = 0.475\log x+0.025 x +0.5x(y-\theta)^2
\end{align*}
and $p(y|\theta)$ is a t-distribution with location parameter $\theta$ and 0.05 degrees of freedom, it follows that the posterior $p_\theta(x|y)$ is a Gamma distribution with parameters $\alpha = 0.525+1, \beta = 0.025+(y-\theta)^2/2$. 

\paragraph{Convexity and Lipschitz continuity}
We further have that
\begin{align*}
     U(\theta_2, x) -  U(\theta_1, x)- \ssp{\nabla_\theta    U(\theta_1, x), \theta_2-\theta_1} = \frac{x}{2}(\theta_2-\theta_1)^2,
\end{align*}
showing that $U$ is only locally smooth w.r.t. the Euclidean distance as $L=x/2$ but convex since $l=0$ (as $x>0$ in this case).
However, $\nabla_x U(\theta, x) = 0.475/x+0.025+0.5(y-\theta)^2$ is not Lipschitz continuous w.r.t. $x$, this causes the ULA update employed in PGD and IPLA to be unstable as shown in Figure~\ref{fig:multimodal}.

\paragraph{Experimental set up}
We set $\theta_0=0$ and $\mu_0(x)=\textrm{Gamma}(x;1, 1)$.
Since $\nabla_x U$ is not Lipschitz continuous, to ensure that PGD and IPLA do not explode we pick $\gamma_n\equiv 0.001$ (larger values of $\gamma_n$ could be used for SMCs-LVM) and iterate for $T=2000$ steps, we fix $N=1000$.

\paragraph{SMC-MML}
\citeapp{johansen2008particle}  uses ideas borrowed from simulated annealing (see, e.g., \citeapp{van1987simulated}) to sample from $\pi^\beta(\theta) \propto \exp(-\beta K(\theta))$, where $K(\theta) =-\log\int_{\real^{d_x}} e^{-U(\theta, x)}d x$ is the marginal log-likelihood, and relies on the fact that as $\beta\to \infty$ the distribution $\pi^\beta$ concentrates on the maximisers of $K$.
 
\citeapp{johansen2008particle} introduces $\beta$ auxiliary copies of the latent variable and considers the extended target distribution $p_\beta(\theta, x_{1:\beta}) \propto\prod_{i=1}^\beta p_\theta(x_i, y)$, which admits $\pi^\beta(\theta)$ as marginal, and builds an SMC sampler targeting $p_\beta$.
The resulting method, named SMC-MML, provides an approximation of the posterior $p_\theta(x|y)$ too.

\subsection{Bayesian logistic regression}
\label{app:blr}

\paragraph{Convexity and Lipschitz continuity}

For this example to negative log-likelihood is given by
\begin{align*}
    U(\theta, x) = (d_x/2)\log(2\uppi) - \sum_{j=1}^{d_y}\left(y_j\log(s(v_j^Tx))+(1-y_j)\log(s(-v_j^Tx))\right)+\frac{\norm{x-\theta}^2}{2},
\end{align*}
from which we obtain
\begin{align*}
    \nabla_\theta  U(\theta, x) &= -( x-\theta).
\end{align*}
If follows that
\begin{align*}
     U(\theta_2, x) -  U(\theta_1, x)- \ssp{\nabla_\theta    U(\theta_1, x), \theta_2-\theta_1} = \frac{1}{2}\norm{\theta_2-\theta_1}^2,
\end{align*}
showing that $U$ is both relatively convex and relatively smooth w.r.t. the Euclidean norm with $l=L=1/2$.
\paragraph{Further numerical results}

Figure~\ref{fig:bayesian_lr_comparison} shows the result of one run of all algorithms when $N=100$, $\gamma_n\equiv 0.001$ and $T=6000$, all algorithms are initialised at $\theta_0=(0, 0, 0)$ and $X_0$ is sampled from a standard normal. We compare the estimated MLE with the true parameter $\theta = (2, 3, 4)$; all algorithms are in agreement altough IPLA returns much noisier results.

\begin{figure}
	\centering
	\begin{tikzpicture}[every node/.append style={font=\normalsize}]
		\node (img1) {\includegraphics[width = 0.4\textwidth]{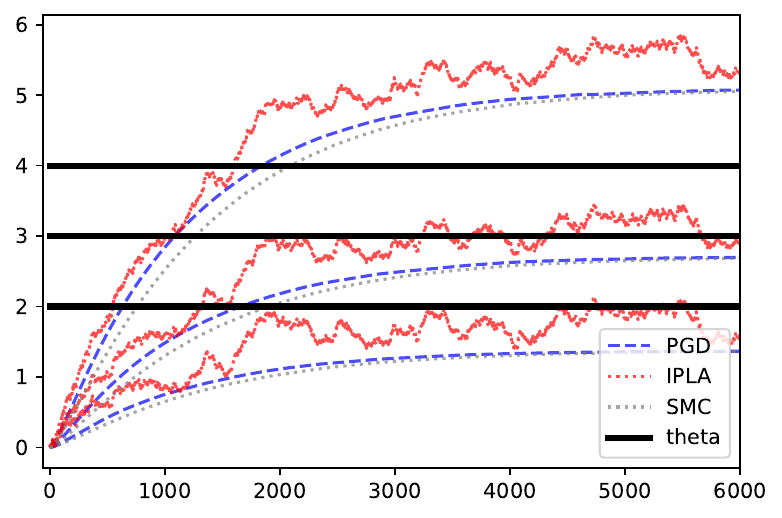}};
		\node[right=of img1, node distance = 0, xshift = -0.5cm] (img2) {\includegraphics[width = 0.4\textwidth]{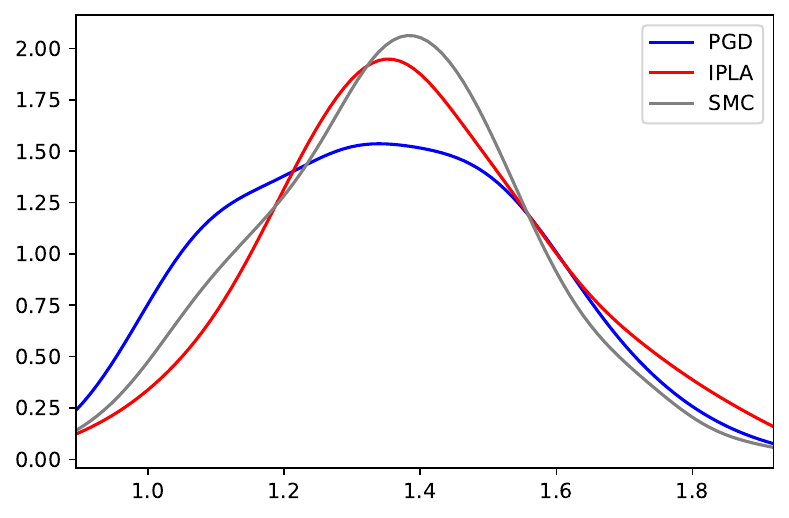}};
  \node[left=of img1, node distance = 0, rotate = 90, anchor = center, yshift = -0.8cm] {$\theta_n$};
    \node[left=of img2, node distance = 0, rotate = 90, anchor = center, yshift = -0.8cm] {$p_{\theta_T^N}(x|y)$};
		\node[below=of img1, node distance = 0, yshift = 1.2cm] {$n$};
		\node[below=of img2, node distance = 0, yshift = 1.2cm] {$x$};
	\end{tikzpicture}
	\caption{Comparison of $\theta$-iterates and first component  of the approximate posterior at the final time for PGD, IPLA and SMCs-LVM with $N=100$, $\gamma = 0.001$ and $T=6000$.}
	\label{fig:bayesian_lr_comparison}
\end{figure}

\begin{table}
\centering
\begin{tabular}{l|cc|cc|cc}
 & \multicolumn{2}{c}{$N=10$} & \multicolumn{2}{c}{$N = 50$}& \multicolumn{2}{c}{$N = 100$}\\
\hline\noalign{\smallskip}
Method & variance & runtime (s) & variance & runtime (s) & variance & runtime (s)\\
\hline\noalign{\smallskip}
PGD & $2.60\cdot10^{-4}$ &\textbf{0.78}& $6.43\cdot10^{-5}$ & 2.85& $2.59\cdot10^{-5}$& \textbf{7.48}\\
IPLA &1.12 &0.80&$1.74\cdot10^{-1}$&\textbf{2.70}& $8.80\cdot10^{-2}$& 7.57\\
SMCs-LVM &$\mathbf{3.20\cdot10^{-5}}$ &4.57&$\mathbf{6.01\cdot10^{-6}}$& 25.76 &$\mathbf{2.54\cdot10^{-6}}$& 50.69\\
\end{tabular}
\caption{Variance of estimates of the first component of $\theta$ for the Bayesian logistic regression model with $N=10, 50, 100$ and their computational times. $\gamma = 0.001, T=6000$ throughout all experiments.  The best values are in bold.}
\label{tab:bayesian_lr2}
\end{table}

\begin{table}
\centering
\begin{tabular}{l|cc|cc|cc}
 & \multicolumn{2}{c}{$N=10$} & \multicolumn{2}{c}{$N = 50$}& \multicolumn{2}{c}{$N = 100$}\\
\hline\noalign{\smallskip}
Method & variance & runtime (s) & variance & runtime (s) & variance & runtime (s)\\
\hline\noalign{\smallskip}
PGD & $1.65\cdot10^{-4}$ &\textbf{0.78}& $3.95\cdot10^{-5}$ & 2.85& $1.29\cdot10^{-5}$& \textbf{7.48}\\
IPLA &1.11 &0.80&$1.71\cdot10^{-1}$&\textbf{2.70}& $8.65\cdot10^{-2}$& 7.57\\
SMCs-LVM &$\mathbf{2.46\cdot10^{-5}}$ &4.57&$\mathbf{4.08\cdot10^{-6}}$& 25.76 &$1.68\cdot10^{-6}$ & 50.69\\
\end{tabular}
\caption{Variance of estimates of the first component of $\theta$ for the Bayesian logistic regression model with $N=10, 50, 100$ and their computational times. $\gamma = 0.001, T=6000$ throughout all experiments.  The best values are in bold. }
\label{tab:bayesian_lr3}
\end{table}

\subsection{Stochastic block model}
\label{app:sbm}
\paragraph{Model specification}

Recall that for an undirected graph with $d_x$ nodes we have $p_\theta(x) = \mathbb{P}(x=q) =p_q$ for $q=1, \dots, Q$. The number of edges is drawn from
$y_{ij}|x_i, x_j \sim \textrm{Bernoulli}(\nu_{x_ix_j})$,
so that $p_\theta(y|x) =\prod_{i, j=1}^{d_x} (1-\nu_{x_ix_j})^{1-y_{ij}}\nu_{x_ix_j}^{y_{ij}}$.
b
The joint negative log-likelihood for this model is
\begin{align*}
    U(\theta, x) &= -\sum_{q=1}^Q\sum_{i=1}^{d_x} \mathbf{1}\{x_i=q\}\log p_q\\ 
    &- \sum_{q, l=1}^Q \sum_{i=1}^{d_x}\sum_{j\neq i}\mathbf{1}\{x_i=q, x_j=l\}\left(y_{ij}\log \nu_{q,l} + (1-y_{ij})\log (1-\nu_{q,l}) \right).
\end{align*}

\paragraph{Convexity and Lipschitz continuity}

The gradients are given by
\begin{align*}
    \nabla_{p_q}U(\theta, x) &= -\sum_{i=1}^{d_x}\mathbf{1}\{x_i=q\}\frac{1}{ p_q}\\
    \nabla_{\nu_{q, l}}U(\theta, x) &= - \sum_{i=1}^{d_x}\sum_{j\neq i}\mathbf{1}\{x_i=q, x_j=l\}\left(\frac{y_{ij}}{\nu_{q,l}} - \frac{1-y_{ij}}{1-\nu_{q,l}} \right).
\end{align*}
If follows that
\begin{align*}
     &U(\theta_2, x) -  U(\theta_1, x)- \ssp{\nabla_\theta    U(\theta_1, x), \theta_2-\theta_1} \\
     &\qquad=\sum_{q=1}^Q\left(\log p_q^{(1)}-\log p_q^{(2)}+\frac{p_q^{(2)}-p_q^{(1)}}{ p_q^{(1)}}\right)\sum_{i=1}^{d_x} \mathbf{1}\{x_i=q\}\\
     &\qquad+ \sum_{q, l=1}^Q \left(\log \nu_{q,l}^{(1)}-\log \nu_{q,l}^{(2)}+\frac{\nu_{q,l}^{(2)}-\nu_{q,l}^{(1)}}{ \nu_{q,l}^{(1)}}\right)\sum_{i=1}^{d_x}\sum_{j\neq i}\mathbf{1}\{x_i=q, x_j=l\}y_{ij}\\
     &\qquad+ \sum_{q, l=1}^Q \left(\log (1-\nu_{q,l}^{(1)})-\log (1-\nu_{q,l}^{(2)})+\frac{\nu_{q,l}^{(2)}-\nu_{q,l}^{(1)}}{ 1-\nu_{q,l}^{(1)}}\right)\sum_{i=1}^{d_x}\sum_{j\neq i}\mathbf{1}\{x_i=q, x_j=l\} (1-y_{ij})
\end{align*}
where we set $\theta_i = \left((p_q^{(i)})_{q=1}^Q, (\nu_{ql}^{(i)})_{q,l=1}^Q)\right)$.

Since $t\mapsto -\log t$ and $t\mapsto -\log (1-t)$ do not have Lipschitz continuous gradients on $[0, 1]$ we conclude that the relative smoothness required by Assumption~\ref{ass:convex} is not satisfied with $h=\norm{\cdot}^2/2$.
However, due to the convexity of $t\mapsto -\log t$ and $t\mapsto -\log (1-t)$ we have that $U(\theta, x)$ is convex (but not strongly) relatively to $h=\norm{\cdot}^2/2$.

As all the components of $\theta$ are constrained to $[0, 1]$ we enforce the constraint using component-wise logarithmic barriers $h(t)= -\log(1-t)-\log t$.

It follows that $\nabla h(t) = 1/(1-t)-1/t$ and  $(\nabla h)^{-1}(t) = (t-2+\sqrt{t^2+4})/(2t)$ and the update for each component of  $\theta$ becomes
\begin{align*}
    \theta_{n+1}(i) &= (\nabla h)^{-1}\left(\frac{1}{1-\theta_n(i)}-\frac{1}{\theta_n(i)} - \gamma_{n+1}\int \nabla_\theta U(\theta_n, x)\mu_n(x)dx\right).
\end{align*}

\bibliographystyleapp{apalike}  
\bibliographyapp{biblio_lvm_md.bib}
\end{document}